%% file: main.tex
\documentclass[sigconf]{acmart}

\pdfoutput=1
\usepackage{booktabs} 
\usepackage{mathrsfs}
\usepackage{mathtools}
\usepackage{array}
\usepackage{amsmath}
\usepackage{amssymb}
\usepackage{color}
\usepackage{graphicx}
\usepackage{refcount}
\usepackage[linesnumbered,ruled,vlined]{algorithm2e}
\makeatletter
\@ifundefined{showcaptionsetup}{}{%
	\PassOptionsToPackage{caption=false}{subfig}}
\usepackage{subfig}
\makeatother

\newcommand{\ie}{\emph{i.e.}}
\newcommand{\eg}{\emph{e.g.}}

\DeclareMathOperator*{\argmin}{arg\,min}
\DeclareMathOperator*{\argmax}{arg\,max}

\newtheorem{assumption}{assumption}
\newtheorem{claim}{Claim}[section]
\newtheorem{property}{Property}

\copyrightyear{2017}
\acmYear{2017}
\setcopyright{rightsretained}
\acmConference{SIGMETRICS '17}{June 5--9, 2017}{Urbana-Champaign, IL, USA}
\acmPrice{}
\acmDOI{http://dx.doi.org/10.1145/3078505.3078529}
\acmISBN{978-1-4503-5032-7/17/06}

\begin{document}
\title{Optimal Posted Prices for Online Cloud Resource Allocation}


\author{Zijun Zhang}
\affiliation{
	\institution{Dept. of Computer Science\\ University of Calgary}
}
\email{zijun.zhang@ucalgary.ca}

\author{Zongpeng Li}
\affiliation{
	\institution{Dept. of Computer Science\\ University of Calgary}
}
\email{zongpeng@ucalgary.ca}

\author{Chuan Wu}
\affiliation{
	\institution{Dept. of Computer Science\\ The University of Hongkong}
}
\email{cwu@cs.hku.hk}

\begin{abstract}
We study online resource allocation in a cloud computing platform, through a {\em posted pricing} mechanism: The cloud provider publishes a unit price for each resource type, which may vary over time; upon arrival at the cloud system, a cloud user either takes the current prices, renting resources to execute its job, or refuses the prices without running its job there. We design pricing functions based on the current resource utilization ratios, in a wide array of demand-supply relationships and resource occupation durations, and prove worst-case competitive ratios of the pricing functions in terms of social welfare. In the basic case of a single-type, non-recycled resource ({\em i.e.}, allocated resources are not later released for reuse), we prove that our pricing function design is {\em optimal}, in that any other pricing function can only lead to a worse competitive ratio. Insights obtained from the basic cases are then used to generalize the pricing functions to more realistic cloud systems with multiple types of resources, where a job occupies allocated resources for a number of time slots till completion, upon which time the resources are returned back to the cloud resource pool.
\end{abstract}

%
%
%



\maketitle

\input{intro}

\input{related}

\input{pricStrat}

\input{genPricStrat}

\input{sim}

\input{conclusion}

\bibliographystyle{ACM-Reference-Format}
\bibliography{sigproc}

\input{appendix}

\end{document}

%% file: intro.tex
\section{Introduction}
Over the past decade, cloud computing has proliferated as the new computing paradigm that provides flexible, on-demand computing services in a pay-as-you-go fashion. Various applications and systems today are built upon cloud computing models, including  
 big data analytics, cloud radio access networks (C-RAN), network function virtualization (NFV), to name a few. Despite the common illusion that a cloud consists of an unlimited `sea' of resources, real-world clouds are constrained by finite system capacity bounds \cite{manvi2014resource,toosi2014interconnected} ({\em e.g.}, physical capacity of a cloud data center), which may become tight in periods of peak demands \cite{buyya2009cloud,garg2013framework}. A fundamental problem in cloud computing is cloud resource allocation, {\em i.e.}, to determine which user demands to satisfy at each time point. A common goal is to maximize the social welfare of the cloud eco-system, which represents the aggregated `happiness' of the cloud provider and the cloud users \cite{an2010automated}. 

Cloud resource allocation in practice exhibits a nature of {\em online decision making}: cloud users with job requests arrive at the cloud system at arbitrary time points, and the cloud provider decides resource allocation upon each job request. 
 A natural, {\em de facto} standard of cloud resource allocation, is through a {\em posted pricing mechanism}: the cloud provider publishes resource prices; cloud users act as price takers who will decide to utilize the resources if the prices are acceptable ({\em i.e.}, its valuation of the job exceeds the cost of resource renting), and will otherwise give up the cloud service.

\textcolor{black}{Major cloud providers today, such as Amazon Web Services, Microsoft Azure, and Google Cloud, typically adopt fixed prices, \ie, resource usage is charged at fixed unit prices posted on their websites. 
However, a dynamic pricing strategy based on realtime demand-supply is more efficient in many scenarios \cite{al2013cloud}, to fully exploit the resource capacity of a cloud system, and to better satisfy user demands. For practical cloud computing systems that employ dynamic pricing strategies, \eg, Amazon EC2 Spot Instances \cite{amazon_ec2si_pricing}, the short-term prices may not be driven by realtime demand-supply \cite{agmon2013deconstructing}; however, the price differences across different service regions and over different time periods are still relevant to demand and supply. Inspired by the Spot Instances model, various dynamic pricing strategies have been proposed in recent literature, including auction mechanisms 
\cite{lin2010dynamic,wang2013revenue,
zaman2013combinatorial,shi2014online, 
zhang2015online,zhou2016efficient,gu2016efficient}, 
and other dynamic pricing strategies for revenue maximization and efficient cloud resource utilization \cite{xu2013dynamic,li2011pricing,mihailescu2010dynamic}.} 

This work studies effective pricing functions for a cloud provider to employ, for computing unit resource prices at each time point. The computed prices are posted as `take it or leave it' prices for cloud users to decide whether to rent the cloud resources (user values not revealed to the cloud provider). Such prices can also serve in a posted-price auction mechanism for cloud job admission and charging. With meticulously designed online prices, our goal is to maximize the social welfare of the cloud, which equals the overall valuation of executed user jobs, minus a possible operational cost, over the entire system span.

\textcolor{black}{While maximizing social welfare does not lead directly to maximizing provider revenue (a natural goal for a cloud provider to pursue), the former is also a very meaningful goal \cite{ma2010resource,nejad2015truthful}. Social welfare represents the aggregate gain of the cloud provider and cloud users, indicating overall system efficiency. Compared to maximizing provider revenue, maximizing social welfare ensures good user experience, which is critical for sustainability of the system in the long run: long-term competitiveness in the market relies on customer happiness, which is instrumental to long-term revenue sustainability of the provider \cite{zhang2013dynamic}. In addition, for public clouds operated by nonprofit organizations, and private clouds for serving internal jobs, maximizing social welfare is more relevant than maximizing revenue \cite{menache2011socially}. In these cases, the pricing schemes studied in this paper can be used as mechanisms for allocating cloud resources to users based on their urgency and priorities. Furthermore, in the auction design literature, there exist techniques that can relate social welfare maximizing mechanisms with revenue maximizing mechanisms \cite{cai2013reducing}.}


Our study of the pricing functions has been partly inspired by dual price design in competitive online algorithms based on the classic primal-dual framework \cite{buchbinder2005online,buchbinder2009design}. In primal-dual online algorithm design, a key idea is to update dual prices using exponential functions for making primal resource allocation decisions, leading to provable competitive ratios. Nonetheless, no explicit justifications were provided in the literature on the choice of using exponential dual price functions.

\textcolor{black}{In this work, we borrow the exponential form of the price function from the literature on primal-dual online algorithms, and propose the optimal form of the exponential pricing functions for a fundamental cloud resource allocation problem. We then provide an intuitive explanation of the optimality of the exponential pricing function. In addition, for the first time in the literature, we generalize the pricing function to scenarios with bounded total demand, where the optimal form is no longer necessarily an exponential function. Interestingly, this result also contributes to the literature on knapsack problems, in that our problem is closely related to a variant of the online knapsack problem \cite{chakrabarty2008online}, where the total weight of items is upper bounded.}

We start by investigating the basic case of a single type of cloud resource without resource recycling, 
and design resource pricing functions based on the current resource utilization levels that capture realtime demand-supply of cloud resources. 
 We prove the optimality of our pricing function design. We then investigate the cases of multiple resource types, and limited resource occupation durations. Our detailed contributions are summarized below.


First, we justify the use of exponential pricing functions in the literature of both cloud computing \cite{shi2014rsmoa,zhang2015online,zhou2016efficient,gu2016efficient,shi2016online} and online algorithms \cite{buchbinder2005online,buchbinder2009design}, both from a theoretical point of view and with intuitive interpretation. We prove the optimality of the pricing function under mild system assumptions that are standard in recent literature. 

Second, we derive the optimal pricing functions for more realistic cloud resource allocation scenarios, where the potential total demand for resources is bounded. 

Third, we extend the pricing functions to take into account multiple resource types. We propose a joint pricing and scheduling strategy when the cloud system runs over multiple time slots. 
We prove tight competitive ratios for these scenarios, which were not properly proven in previous literature. We make no assumptions on the arrival process and the distribution of user valuations.

In addition, we further verify effectiveness of our price design in realistic cloud computing scenarios using simulation studies, relaxing assumptions made in the theoretical analysis. 
We show that the parameters involved in our pricing functions can be practically optimized in different scenarios, to achieve consistently good performance ratios, as compared to the offline optimal social welfare. 
 
Finally, we note that our pricing models and algorithms are generally applicable to posted pricing mechanism design in other online resource allocation systems, which share similar characteristics as a cloud computing system.

In the rest of the paper, we review related literature in Sec.~\ref{sec:related}. The basic and general models of cloud resource pricing are studied in Sec.~\ref{sec:pricing-basic} and Sec.~\ref{sec:pricing-general}, respectively.  Sec.~\ref{sec:sim} presents simulation studies, and Sec.~\ref{sec:conclusion} concludes the paper. 

%% file: related.tex
\section{Related Work}
\label{sec:related}
Recently, auction mechanisms have been extensively studied for online cloud resource allocation and pricing. \citeauthor{zhang2015online} \cite{zhang2015online} design an online auction mechanism for IaaS clouds, aiming to maximize both social welfare and provider profit. \citeauthor{zhou2016efficient} \cite{zhou2016efficient} extend the auction mechanism to deal with computing jobs with soft deadlines. \citeauthor{shi2016online} \cite{shi2016online} propose an online mechanism for virtual cluster allocation and pricing. 
 These studies exploit the primal-dual framework for online mechanism design, and use exponential pricing functions to compute dual prices, which decide resource allocation and user payments. Competitive ratios of the online mechanisms are proven, but the rational of adopting exponential pricing functions is lacking, and the optimality of such exponential functions are not studied. Indeed, a wide spectrum of increasing functions are conceivable for cloud resource pricing. Our pricing functions are applicable to both posted pricing mechanisms and online auctions. The analysis of optimality of our pricing functions is independent from the primal-dual framework.  

\textcolor{black}{Apart from auction mechanisms, a wide range of resource pricing schemes have been studied in the literature. While static pricing schemes are prevalent in today's cloud computing market, dynamic pricing schemes based on realtime demand-supply are shown to be more efficient in many scenarios \cite{al2013cloud}. \citeauthor{li2011pricing} \cite{li2011pricing} design a pricing algorithm for cloud resources, which analyses the historical utilization ratio of the resource, and updates current prices accordingly. Their experiment demonstrates the advantage of the pricing algorithm in terms of cost reduction and efficient resource allocation. \citeauthor{mihailescu2010dynamic} \cite{mihailescu2010dynamic} propose a dynamic pricing scheme for federated clouds, where different cloud providers share and trade resources for enhanced scalability and reliability. They show that user welfare and the percentage of successful requests are increased by dynamic pricing, as compared to fixed pricing. The pricing schemes developed in this work are both dynamic and usage-based, \ie, the unit price of cloud resource is driven by demand-supply dynamics, and the total price is proportional to the amount and service time of requested resources.}

The online social welfare maximization problem studied in this work related to a variant of the online knapsack problem \cite{chakrabarty2008online}. Two assumptions are made in this literature: the weight of each item is much smaller than the capacity of the knapsack, and the density (value to weight ratio) of every item falls in a known range $\left[L,U\right]$. Under these assumptions, \citeauthor{buchbinder2005online} \cite{buchbinder2005online,buchbinder2006improved} design an algorithm achieving a competitive ratio of $O\left(\log\left(U/L\right)\right)$, as well as an $\Omega\left(\log\left(U/L\right)\right)$ lower bound on the competitive ratio of any algorithm. In the context of advertising auctions, \citeauthor{zhou2008budget} \cite{zhou2008budget} design a $\left(\log\left(U/L\right)+1\right)$-competitive algorithm for an online knapsack problem under the above assumptions. Interestingly, their algorithm is equivalent to our proposed pricing strategy for the most basic case, as will be discussed in Sec.~\ref{sec:arbitDemand}. Nevertheless, our proof of optimality is different from that given by \citeauthor{zhou2008budget} \cite{zhou2008budget}, and leads to an intuitive interpretation on the choice of exponential pricing functions. \textcolor{black}{More importantly, the total weight of items is assumed to be unbounded in the previous work, which is hardly the case for any real-world applications. In this work, we develop a more general pricing strategy that achieves better competitive ratios for bounded total weight, and we prove the optimality of the proposed strategy.}

%% file: pricStrat.tex
\section{Pricing for Cloud Resource Allocation: the Basic Case}
\label{sec:pricing-basic}

\begin{table}
	\caption{Notation and definition}
	\begin{tabular}{|c|m{180pt}|}
		\hline
		$\mathcal{U}$ & set of users\\
		\hline 
		$\mathcal{R}$ & set of resource types\\
		\hline 
		$\mathcal{T}$ & set of all time slots\\
		\hline 
		$\mathcal{T}_{i}$ & set of time slots required by user $i$\\
		\hline
		$d_{i,r}$ & amount of resource $r$ demanded by user $i$\\
		\hline
		$d_{i}$ & total amount of resource demanded by user $i$\\
		\hline
		$v_{i}$ & value of successfully finishing user $i$'s job\\
		\hline
		$p$ & unit resource price at the time of user arrival\\
		\hline
		$\underline{p}/\overline{p}$ & lower/upper bound of $v_{i}/d_{i}$\\
		\hline
		$\gamma$ & ratio between $\overline{p}$ and $\underline{p}$\\
		\hline
		$\rho$ & resource utilization level\\
		\hline
		$\rho_r\mbox{*}$ & final resource utilization level as defined by Definition \ref{def:finRho}\\
		\hline
		$\beta$ & scarcity level as defined by Definition \ref{def:limitedBeta}\\
		\hline
		$V_{ol}\left(\rho\mbox{*}\right)$ & total value obtained by an online solution, given a final utilization level $\rho\mbox{*}$\\
		\hline
		$V_{opt}\left(\rho\mbox{*}\right)$ & total value obtained by an optimal offline solution, given a final utilization level $\rho\mbox{*}$\\
		\hline 
	\end{tabular}
\end{table}

In this section, we start by designing pricing functions for a basic, yet fundamental version of the online resource allocation problem, following the posted pricing framework as described in Algorithm \ref{alg:postedPrice}. 

\begin{algorithm}
	\KwIn{$d_{i},v_{i},\forall i\in\mathcal{U}$}
	\KwOut{$x_{i},\forall i\in\mathcal{U}$}
	$\rho=0$ \tcp*{Initialize the resource utilization}
	\For{$i\in \mathcal{U}$}{
		\tcc{Upon the arrival of each user $i$}
		\uIf{$v_{i}\geq d_{i}P\left(\rho\right)$ \textbf{and} $\rho+d_{i}\leq 1$}{
			\tcc{User $i$ accepts the posted price}
			$x_{i}=1$\;
			$\rho=\rho+d_{i}$ \tcp*{Allocate resource to user $i$}
		}\Else{
			\tcc{User $i$ rejects the posted price}
			$x_{i}=0$\;
		}
	}
	\caption{Online pricing and resource allocation \label{alg:postedPrice}}
\end{algorithm}

\subsection{The Basic Resource Allocation Problem}
Consider a cloud provider whose data center is for now assumed to provision a single type of resource. The resource is to be allocated to a large number of cloud users. The users in a set $\mathcal{U}$ come in an arbitrary sequence. Upon arrival, a user decides immediately whether to rent some of the cloud resources, by comparing the valuation of its job with the overall price of required resources for executing the job. Let $d_{i}$ denote the amount of resource demanded by a user $i\in\mathcal{U}$, and $v_{i}$ be the value of successfully finishing $i$'s job. \textcolor{black}{A user may decide $v_{i}$ according to different factors, such as the purpose and priority of the job, and how the user can gain from the job completion}. Without loss of generality, we normalize user resource demands, assuming the total amount of resource in the cloud is $1$, so that $d_{i}$ can be considered as the proportion of the entire resource pool demanded by user $i$. Let $p$ be the unit price of the resource posted by the cloud provider, which may vary over time. A user $i$ accepts the price and rents resource at quantity $d_{i}$, if and only if $v_{i}\geq d_{i}p$, where $p$ is the current unit resource price at the time of user arrival. \textcolor{black}{To put it another way, $v_{i}$ can be simply seen as a threshold for whether a price is acceptable to user $i$.} 
 In this section, we assume that each unit of the resource, once allocated, will not be returned to the resource pool. 
 
The utility of the cloud provider is the total payment received. The utility of a served user is the valuation of its job minus its payment. The utility of an unserved user is zero. Since payments cancel themselves in the summation, the social welfare of the entire cloud system, including utilities of both the cloud provider and the cloud users, is equivalent to the total valuation of jobs that are served, assuming no operational cost of the cloud. 

Let $x_{i}$ indicate whether user $i$ rents resource (at quantity $d_i$) or not upon its arrival. The social welfare maximization problem can be formulated as an integer linear program (ILP):

	\begin{equation}\label{eq:prob1}
		\textrm{maximize}\quad \sum_{i\in\mathcal{U}}v_{i}x_{i}
	\end{equation}
	\quad\text{s.t.:}
	\begin{align*}
		\sum_{i\in\mathcal{U}}d_{i}x_{i}\leq 1 & \quad\quad\quad\quad\quad\quad\quad (\ref{eq:prob1}a) \\
		x_{i}\in\left\{0,1\right\},&\forall i\in\mathcal{U} \quad\quad\quad\quad (\ref{eq:prob1}b)
	\end{align*}

\noindent This is a 0-1 knapsack problem, and can be solved exactly using dynamic programming in the offline setting. However, for the online problem we are investigating, the columns of the coefficient matrix of constraint (\ref{eq:prob1}a), corresponding to different online-arriving users, are revealed one-by-one, while the value of $x_{i}$ is to be determined immediately when a user comes to the cloud. We apply an online resource allocation algorithm, as shown in Algorithm \ref{alg:postedPrice}, 
to decide resource allocation given resource prices.

The performance of the posted pricing mechanisms in the online resource allocation algorithm clearly depends on the pricing function. Practically, we do not assume that users reveal their job valuations to the cloud provider. 
Consequently, the pricing strategy depends only on the demand-supply relationship of cloud resources. We will use the standard notion of {\em competitive ratio} to evaluate the quality of our online resource allocation solution, which is defined as the ratio between the optimal objective value of the offline problem \eqref{eq:prob1} and that of the online solution. The smaller (closer to 1) the competitive ratio is, the better the online resource allocation solution. More specifically, we will focus on the {\em worst-case} competitive ratio (as opposed to average-case competitive ratio). 
We first make the following two mild assumptions:

\begin{assumption}\label{asmp:varCstr}
	The variability of users' valuations is constrained, \ie, $\underline{p}\leq v_{i}/d_{i}\leq\overline{p},\forall i\in\mathcal{U}$, where $\underline{p}$ and $\overline{p}$ are lower bound and upper bound of the per-unit-resource job valuation of all users, respectively.
\end{assumption}
\begin{assumption}\label{asmp:smallDemand}
	The resource demand of each user is much smaller than the total resource capacity, \ie, $d_{i}\ll 1,\forall i\in\mathcal{U}$. 
\end{assumption}

\textcolor{black}{Assumption~\ref{asmp:smallDemand} is reasonable when considering large-scale data centers, where the total resource capacity refers to that of the entire data center. We make this assumption mainly to facilitate our theoretical analysis, 
 such that techniques from calculus (differentiation) can be used, and very extreme cases can be eliminated that are rare in practice. For example, if a high-valued bid demanding almost all the resource from a cloud provider is rejected, because a small fraction of the resource is occupied by other users, then the worst-case competitive ratio can be infinitely large. In addition, such an assumption is standard in the literature of online resource allocation \cite{zhang2015online, zhou2016efficient} and online knapsack problems \cite{buchbinder2005online,buchbinder2006improved,chakrabarty2008online,zhou2008budget}. }

\textcolor{black}{Nonetheless, it is also possible to relax Assumption~\ref{asmp:smallDemand} to specifying an upper bound on $d_{i}$ instead, without significantly affecting our theoretical result. Specifically, we can use difference equation and summation, instead of differential equation and integration, to derive similar results. In addition, we will relax this assumption completely in our empirical studies.}

\subsection{Pricing Function Design}
\textcolor{black}{We design pricing functions that adjust resource prices based on realtime demand-supply. To this end, it is helpful to have some prior knowledge about the total resource demand. In practice, unlimited total resource demand is rare; an estimated upper bound on the overall resource demand can often be obtained. This is reflected through the following definition.}
\begin{definition}\label{def:limitedBeta}
	Suppose the total resource demand of all users is upper bounded by $1+\beta$ times the total resource supply, {\em i.e.}, $\sum_{i\in\mathcal{U}}d_{i}\le 1+\beta$, with $\beta>-1$. We refer to $\beta$ as the {\em scarcity level} of the resource.
\end{definition}

\textcolor{black}{It is possible to have a known lower bound on the overall resource demand as well, but our algorithm design and analysis do not rely on such a lower bound.}

\textcolor{black}{We next present the optimal pricing function for $\beta\rightarrow \infty$, and then derive the optimal pricing functions for finite $\beta$, based on the insight we gain from the analysis of the first case. We then further show that the results can be extended to the case that linear operational costs of cloud resources are considered in Sec.~\ref{sec:linOpCost}.}

\subsubsection{Pricing Function for Large Total Demand}\label{sec:arbitDemand}
We begin with the case that the total demand for resource is much larger than the capacity of the cloud resource pool. We propose an optimal pricing function for the case that $\beta\rightarrow \infty$, and then show the same pricing function is in fact optimal as long as $\beta\ge 1$ ({\em i.e.}, the upper bound on the overall resource demand is at least twice of the resource capacity). 
\textcolor{black}{
\begin{definition}\label{def:optFcn}
In Algorithm \ref{alg:postedPrice}, oblivious of true valuations of users, a pricing function is {\em optimal} if it achieves the smallest possible worst-case competitive ratio in social welfare under Assumptions~\ref{asmp:varCstr} and \ref{asmp:smallDemand}.
\end{definition}
}

Let $\rho$ be the resource utilization level, {\em i.e.}, the amount of the resource already allocated. \textcolor{black}{Note that $\rho$ is a function of time, but this dependency is omitted for notational simplicity.} The unit price of the resource at the respective resource utilization level is denoted by 
 $P_{1}\left(\rho\right)$, designed as follows: 
\begin{equation}
	P_{1}\left(\rho\right)=\begin{cases}
		\underline{p}, & \rho\in\left[0,1/\left(\log\gamma+1\right)\right]\\
		\underline{p}e^{\left(\log\gamma+1\right)\rho-1}, & \rho\in\left(1/\left(\log\gamma+1\right),1\right)\\
		+\infty, & \rho=1
	\end{cases},\label{eq:priFcn1}
\end{equation}

\noindent where $\gamma=\overline{p}/\underline{p}$. An illustration of the pricing function for $\underline{p}=1$, $\overline{p}=10$ is given in Fig.~\ref{fig:priFcn1} (blue lines in both subfigures). Intuitively, when $\rho$ is quite small, it is desirable to keep the price at the lowest level ($\underline{p}$), to allow all potential users to rent the resource. 
 As $\rho$ increases, the amount of satisfied demand increases, as well as the obtained social welfare, and hence it is reasonable to raise the price to filter out users with low valuations. When $\rho=1$, the resource is exhausted, so we use an infinitely high price to reject all subsequent users.
Note that even if we need the lower bound and upper bound of the per-unit-resource valuation in (\ref{eq:priFcn1}), when applying this pricing function in online resource allocation, we can use estimates of the bounds, which can be further calibrated over time when more users have arrived and more user price taking decisions are learned.

\begin{figure}[!t]
	\begin{centering}
		\subfloat[The pricing function and competitive ratio.]{
			\includegraphics[width=0.45\textwidth]{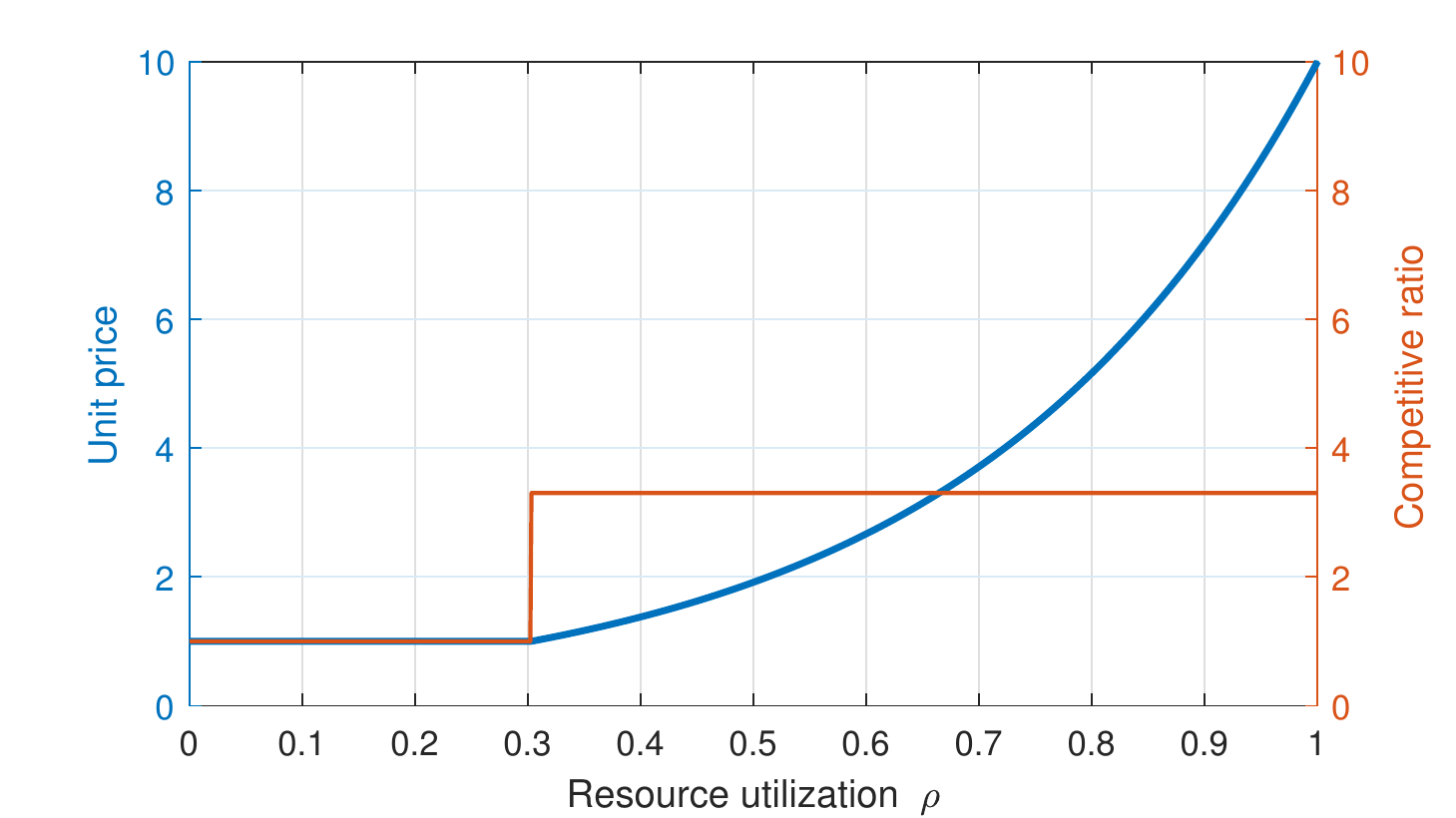}
			\label{fig:priFcn1_curv}}
		\par\end{centering}
	\vspace*{-10pt}
	\begin{centering}
		\subfloat[The worst-case $V_{ol}(\rho\mbox{*})$ and $V_{opt}(\rho\mbox{*})$
		for $\rho\mbox{*}=0.7$ visualized by AUCs. Each AUC indicates the total value obtained by an online or offline solution. ]{
			\includegraphics[width=0.45\textwidth]{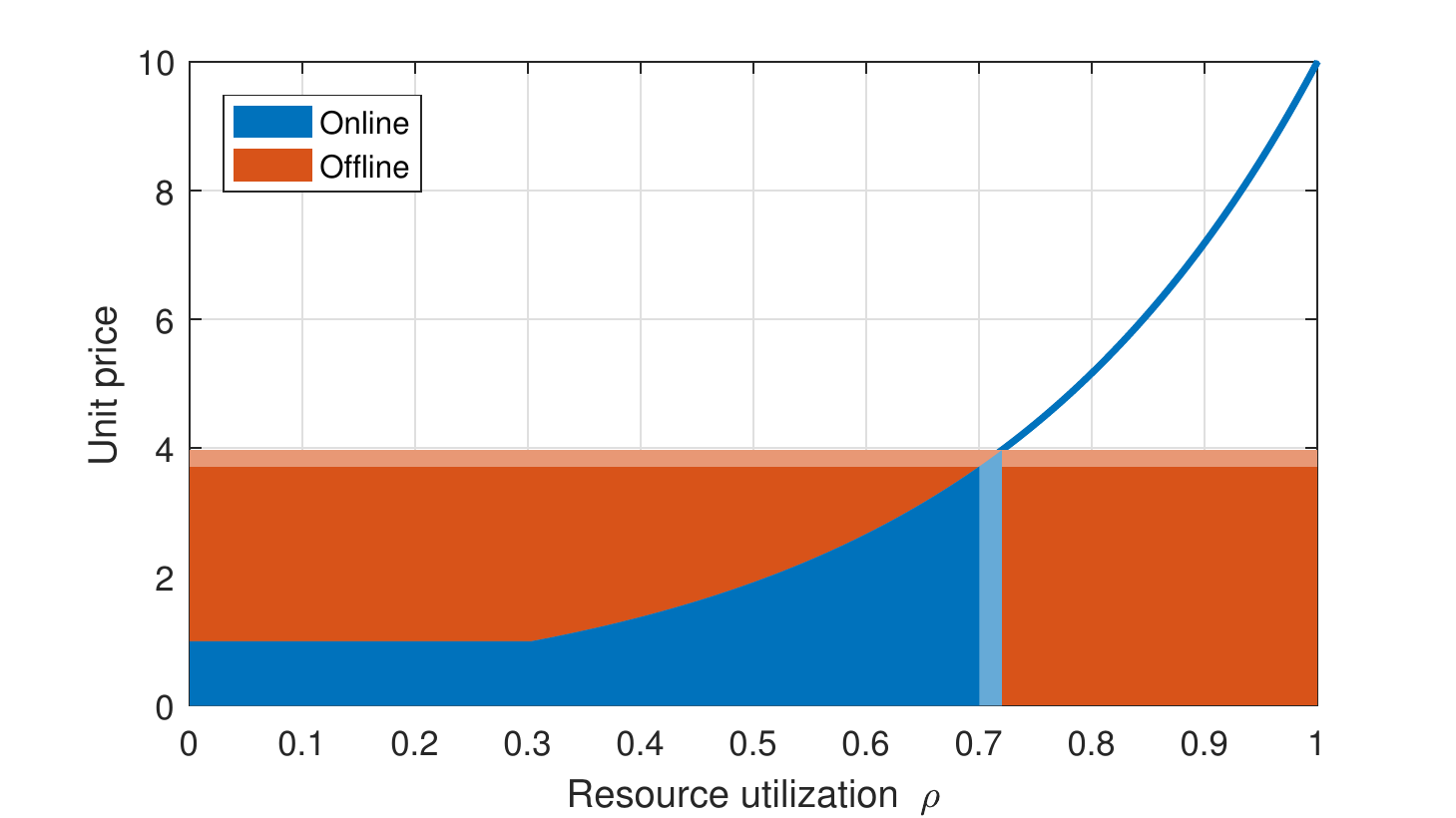}
			\label{fig:priFcn1_auc}}
		\par\end{centering}
	\vspace*{0pt}
	\caption{An illustration of pricing function (\ref{eq:priFcn1}) for $\underline{p}=1$, $\overline{p}=10$.}
	\label{fig:priFcn1}
\end{figure}

We next prove the worst-case competitive ratio of Algorithm \ref{alg:postedPrice} achieved when using the pricing function in \eqref{eq:priFcn1}, as well as the optimality of the pricing function when $\beta\rightarrow \infty$ (this default condition omitted in all lemmas, claim and theorems before Theorem \ref{thm:betaGeq_1}), and then generalize the conclusion to the case $\beta\ge 1$ in Theorem \ref{thm:betaGeq_1}. 

\begin{definition}\label{def:finRho}
	$\rho\mbox{*}\in\left[0,1\right]$ denotes the {\em final} utilization level of the resource 
	after all users have decided whether to rent the cloud resource to execute their jobs. 
\end{definition}

The following lemma implies when the final resource utilization level is low, the total demand of potential users also tends to be low, thus it is possible to satisfy all user demand online. 

\begin{lemma}\label{lem:priFcn1_const_CR}
	If $\rho\mbox{*}\in\left[0,1/\left(\log\gamma+1\right)\right]$, the worst-case competitive ratio achieved by Algorithm \ref{alg:postedPrice} using the pricing function in \eqref{eq:priFcn1} is $\alpha_{1,1}=1$.
\end{lemma}
\begin{proof}
	According to the pricing function in \eqref{eq:priFcn1}, for $\rho\in[0,\allowbreak1/\left(\log\gamma+1\right)]$, the unit price is a constant, $\underline{p}$, which by Assumption~\ref{asmp:varCstr} is acceptable to any potential user, thus $\rho\mbox{*}\in\left[0,1/\left(\log\gamma+1\right)\right]$ implies that the total demand of all users is exactly $\rho\mbox{*}$. The social welfare achieved by the pricing function in \eqref{eq:priFcn1} is the total value of all users, which is also the maximum possible social welfare achieved by solving the offline problem \eqref{eq:prob1}. Therefore, the worst-case competitive ratio is $1$.
\end{proof}
 
For a final utilization level $\rho\mbox{*}$, we let $V_{ol}(\rho\mbox{*})$ be the total value obtained by an online solution, and $V_{opt}(\rho\mbox{*})$ be that obtained by an optimal offline solution. Thus, in any worst case, the ratio $V_{opt}(\rho\mbox{*})/V_{ol}(\rho\mbox{*})$ is maximized.
 
\begin{lemma}\label{lem:priFcn1_nconst_CR}
	If $\rho\mbox{*}\in\left(1/\left(\log\gamma+1\right),1\right]$, the worst-case competitive ratio achieved by Algorithm \ref{alg:postedPrice} using the pricing function in \eqref{eq:priFcn1} is $\alpha_{1,2}=\log\gamma+1$. 
\end{lemma}
\begin{proof}
	For any $\rho\mbox{*}\in\left(1/\left(\log\gamma+1\right),1\right]$, the worst case of the online solution is that the valuations of satisfied users are the same as the prices they accept. By Assumption~\ref{asmp:smallDemand}, the minimum total value of an online solution is
	\begin{equation}\label{eq:V_ol1}
		V_{ol}\left(\rho\mbox{*}\right)=\int_{0}^{\rho\mbox{*}}P_{1}\left(\rho\right)d\rho=\frac{\underline{p}}{\log\gamma+1}e^{\left(\log\gamma+1\right)\rho\mbox{*}-1},
	\end{equation}
	as shown by the blue area under the curve (AUC) in Fig.~\ref{fig:priFcn1_auc}. At the same time, any unsatisfied user has a unit value smaller than $P_{1}\left(\rho\mbox{*}\right)$, because otherwise $\rho\mbox{*}$ cannot be the final resource utilization. Hence in the worst case, there can be a set of unsatisfied users with a total demand of $1$ (\ie, $\sum_{i\in\mathcal{U}_{opt}}d_{i}=1,\forall r\in\mathcal{R}$, where $\mathcal{U}_{opt}$ is the set of user chosen by the optimal offline solution), and each with a unit value of $P_{1}\left(\rho\mbox{*}\right)-\epsilon_{i}$, where $\epsilon_{i}$ is an arbitrarily small positive number, such that the optimal offline solution is to satisfy their demands with all available resource. This yields the maximum optimal offline total value given Eq.~\eqref{eq:V_ol1}:
	\begin{equation}\label{eq:V_opt1}
		\begin{split}
		V_{opt}\left(\rho\mbox{*}\right)= & \sum_{i\in\mathcal{U}_{opt}}d_{i}\left(p\left(\rho\mbox{*}\right)-\epsilon_{i}\right)=p\left(\rho\mbox{*}\right)-\epsilon\\
		= & ~\underline{p}e^{\left(\log\gamma+1\right)\rho\mbox{*}-1}-\epsilon,
		\end{split}
	\end{equation}
	as shown by the red AUC (partially covered by the blue one) in Fig.~\ref{fig:priFcn1_auc}. Here, $\epsilon=\sum_{i\in\mathcal{U}_{opt}}\epsilon_{i}$, and hence can also be arbitrarily small. Note that, there can be a case which leads to a larger optimal offline total value, by increasing the online value corresponding to $\rho\in\left[0,\rho\mbox{*}\right]$ (\ie, the blue AUC in Fig.~\ref{fig:priFcn1_auc}) until it is large enough and becomes part of the optimal offline value. However, the online value will increase more than the optimal offline value does in this case, making it impossible to be a worst case. Therefore, the worst-case competitive ratio $\ensuremath{\alpha_{1,2}=\sup_{\epsilon>0}\frac{V_{opt}\left(\rho\mbox{*}\right)}{V_{ol}\left(\rho\mbox{*}\right)}=\log\gamma+1},\forall\rho\mbox{*}\in\left(1/\left(\log\gamma+1\right),1\right]$.
\end{proof}

An illustration of the worst-case competitive ratio at different final resource utilization levels is shown in Fig.~\ref{fig:priFcn1_curv} (red line).

\begin{theorem}\label{thm:priFcn1_cr}
	The worst-case competitive ratio of Algorithm \ref{alg:postedPrice} using the pricing function in \eqref{eq:priFcn1} is
	\begin{equation}
	\alpha_{1}=\log\gamma+1.\label{eq:alpha_1}
	\end{equation}
\end{theorem}
\begin{proof}
	The worst-case competitive ratio of the pricing function in \eqref{eq:priFcn1} is the maximum possible competitive ratio for all $\rho\mbox{*}\in\left[0,1\right]$. Hence following Lemma~\ref{lem:priFcn1_const_CR} and \ref{lem:priFcn1_nconst_CR}, $\alpha_{1}=\max\left\{\alpha_{1,1},\alpha_{1,2}\right\}$ $=\log\gamma+1$. \end{proof}

We next show the optimality of the pricing function based on the observation that, to achieve a finite worst-case competitive ratio, any pricing function should contain a constant ($\underline{p}$) part at the beginning of the function. 

\begin{claim}\label{clm:constFcn1}
	If a pricing function $P\left(\rho\right)$ achieves a finite worst-case competitive ratio of $\alpha$, then $P\left(\rho\right)=\underline{p}, \forall\rho\in\left[0,1/\alpha\right]$.
\end{claim}
\begin{proof}
	If the claim does not hold and $P\left(0\right)>\underline{p}$, there can be a case where $\rho\mbox{*}=0$, such that the online total value $	V'_{ol}\left(\rho\mbox{*}\right)=0$, while the optimal offline total value $V'_{opt}\left(\rho\mbox{*}\right)=P\left(0\right)-\epsilon>0$, where $\epsilon$ is an arbitrarily small positive number. Thus the worst-case competitive ratio $\alpha\geq\sup_{\epsilon>0}\frac{V'_{opt}\left(\rho\mbox{*}\right)}{V'_{ol}\left(\rho\mbox{*}\right)}=+\infty$, which contradicts the assumption that $\alpha$ is finite.
	
	If the claim does not hold and $P\left(0\right)=\underline{p}$, there must be a $\rho_{0}\in\left(0,1/\alpha\right]$ such that $P\left(\rho_{0}\right)>P\left(\rho\right), \forall\rho\in\left[0,\rho_{0}\right)$. There can be a case where $\rho\mbox{*}=\rho_{0}$, such that the online total value
	\[
	V'_{ol}\left(\rho\mbox{*}\right)=\int_{0}^{\rho_{0}}P\left(\rho\right)d\rho<\rho_{0}P\left(\rho_{0}\right),
	\]
	while the optimal offline total value $V'_{opt}\left(\rho\mbox{*}\right)=P\left(\rho_{0}\right)-\epsilon$. Thus the worst-case competitive ratio $\alpha\geq\sup_{\epsilon>0}\frac{V'_{opt}\left(\rho\mbox{*}\right)}{V'_{ol}\left(\rho\mbox{*}\right)}>1/\rho_{0}$, which contradicts $\rho_{0}\leq1/\alpha$.
\end{proof}

\textcolor{black}{
\begin{theorem}\label{thm:priFcn1_opt}
	the pricing function in \eqref{eq:priFcn1} is optimal according to Definition~\ref{def:optFcn}, {\em i.e.}, using it Algorithm \ref{alg:postedPrice} achieves the smallest worst-case competitive ratio.
\end{theorem}}
\begin{proof}
	We prove this theorem by way of contradiction. Assume that there exists a pricing function, $P'_{1}\left(\rho\right)$, which achieves a worst-case competitive ratio $\alpha'_{1}<\alpha_{1}$. According to Claim~\ref{clm:constFcn1} and Theorem~\ref{thm:priFcn1_cr}, we have $P'_{1}\left(\rho\right)=\underline{p}, \forall\rho\in\left[0,1/\alpha'_{1}\right]$, and hence
	\[
	\int_{0}^{1/\alpha'_{1}}P'_{1}\left(\rho\right)d\rho<\int_{0}^{1/\alpha'_{1}}P_{1}\left(\rho\right)d\rho,
	\]
	where $P_{1}\left(\rho\right)$ is the pricing function in \eqref{eq:priFcn1}.
	
	If there exists some $\rho\in\left(1/\alpha'_{1},1\right)$ such that $P'_{1}\left(\rho\right)\geq P_{1}\left(\rho\right)$ we find the smallest one, and denote it by $\rho_{1}$. Then there can be a case where $\rho\mbox{*}=\rho_{1}$, such that the online total value
	\[
	V'_{ol}\left(\rho\mbox{*}\right)=\int_{0}^{\rho_{1}}P'_{1}\left(\rho\right)d\rho<\int_{0}^{\rho_{1}}P_{1}\left(\rho\right)d\rho=V_{ol}\left(\rho\mbox{*}\right),
	\]
	while the optimal offline total value $V'_{opt}\left(\rho\mbox{*}\right)=P'_{1}\left(\rho_{1}\right)-\epsilon\geq P_{1}\left(\rho_{1}\right)-\epsilon=V_{opt}\left(\rho\mbox{*}\right)$, where $\epsilon$ is an arbitrarily small positive number. Thus the worst-case competitive ratio $\alpha'_{1}\geq\sup_{\epsilon>0}\frac{V'_{opt}\left(\rho\mbox{*}\right)}{V'_{ol}\left(\rho\mbox{*}\right)}>\sup_{\epsilon>0}\frac{V_{opt}\left(\rho\mbox{*}\right)}{V_{ol}\left(\rho\mbox{*}\right)}=\alpha_{1}$, contradicting the assumption $\alpha'_{1}<\alpha_{1}$. Therefore, $P'_{1}\left(\rho\right)<P_{1}\left(\rho\right),\forall\rho\in\left(1/\alpha'_{1},1\right)$.
	
	For $\rho\mbox{*}=1$, since $P'_{1}\left(1\right)\leq\overline{p}$ (a unit price higher than $\overline{p}$ will have all potential users rejected) is finite, we now have
	\[
	V'_{ol}\left(\rho\mbox{*}\right)=\int_{0}^{1}P'_{1}\left(\rho\right)d\rho<\int_{0}^{1}P_{1}\left(\rho\right)d\rho=V_{ol}\left(\rho\mbox{*}\right).
	\]
	However, as the resource is exhausted, subsequent users will not be served, regardless of their valuations. There can be a case where the optimal offline total value $V'_{opt}\left(\rho\mbox{*}\right)=\overline{p}=V_{opt}\left(\rho\mbox{*}\right)$. Thus the worst-case competitive ratio $\alpha'_{1}\geq\frac{V'_{opt}\left(\rho\mbox{*}\right)}{V'_{ol}\left(\rho\mbox{*}\right)}>\frac{V_{opt}\left(\rho\mbox{*}\right)}{V_{ol}\left(\rho\mbox{*}\right)}=\alpha_{1}$, contradicting the assumption that $\alpha'_{1}<\alpha_{1}$.
\end{proof}

We next generalize the optimality result for all $\beta\geq 1$.
\textcolor{black}{
\begin{theorem}\label{thm:betaGeq_1}
	For $\beta\geq 1$, the pricing function in \eqref{eq:priFcn1} is optimal according to Definition~\ref{def:optFcn}, and the corresponding worst-case competitive ratio is $\alpha_{1}$.
\end{theorem}}
\begin{proof}
	For any possible input set of users, we can prune the users that can neither be satisfied by the online solution, nor by the optimal offline solution, without affecting the online or offline social welfare, given a certain pricing function. Clearly, the resulting set of users has a total demand no greater than $2$, which can also happen given any $\beta\geq 1$. Consequently, all the discussions above 
	 can be generalized to $\beta\geq 1$.
\end{proof}

The following property (which holds for all $\beta\geq 1$) is useful for guiding the design of pricing functions in more realistic cloud computing scenarios.

\begin{property}\label{prop:stableCR}
	For the pricing function in \eqref{eq:priFcn1}, and any $\rho\mbox{*}\in\left(1/\alpha_{1},1\right]$, \ie, the monotonically increasing part of $P_{1}\left(\rho\right)$, we have
	\begin{equation}
	\sup_{\epsilon>0}V_{opt}\left(\rho\mbox{*}\right)=\alpha_{1}V_{ol}\left(\rho\mbox{*}\right),\label{eq:stableCR}
	\end{equation}
	and hence
	\begin{equation}
	\frac{d\sup_{\epsilon>0}V_{opt}\left(\rho\mbox{*}\right)}{d\rho\mbox{*}}=\alpha_{1}\frac{dV_{ol}\left(\rho\mbox{*}\right)}{d\rho\mbox{*}},\label{eq:stableCR_diff}
	\end{equation}
	and a constant (w.r.t. $\rho\mbox{*}$) worst-case competitive ratio, $\alpha_{1}$.
\end{property}
\begin{proof}
	The proof is a corollary that follows from Eq.~\eqref{eq:V_ol1}, \eqref{eq:V_opt1}, Lemma \ref{lem:priFcn1_nconst_CR} and Theorem~\ref{lem:priFcn1_nconst_CR}.
\end{proof}

Property~\ref{prop:stableCR} is illustrated in Fig.~\ref{fig:priFcn1_auc}, where the light red area 
corresponds to $\Big.\frac{dV_{opt}\left(\rho\mbox{*}\right)}{d\rho\mbox{*}}\Big|_{\rho\mbox{*}=0.7}$, and the light blue area corresponds to $\Big.\frac{dV_{ol}\left(\rho\mbox{*}\right)}{d\rho\mbox{*}}\Big|_{\rho\mbox{*}=0.7}$. Intuitively, this property implies the best trade-off between the worst-case competitive ratios corresponding to different $\rho\mbox{*}$ values. That is, any changes to the pricing function in \eqref{eq:priFcn1} that may decrease the competitive ratio for some $\rho\mbox{*}$, will unavoidably increase the competitive ratio for some other $\rho\mbox{*}$, and thus can only lead to a worse competitive ratio over all possible values of $\rho\mbox{*}$. 

\subsubsection{Pricing Function for Small Total Demand}
In the case that $\beta\in\left(-1,0\right]$, the total resource demand is no larger than the total resource supply. The optimal strategy is simply serving all user demands by setting a unit resource price below the smallest per-unit-resource valuation of cloud users. 
\textcolor{black}{
\begin{theorem}\label{thm:betaLeq_1}
For $\beta\in\left(-1,0\right]$, pricing function 
\begin{equation}\label{eq:priFcn4}
P_{4}\left(\rho\right)=\underline{p}
\end{equation}
is optimal according to Definition~\ref{def:optFcn}, and the corresponding worst-case competitive ratio achieved by Algorithm \ref{alg:postedPrice} is $1$.
\end{theorem}}

The proof is straightforward and hence omitted.

\subsubsection{Pricing Function for Total Demand Up to Twice of Supply}
In the case that $\beta\in\left(0,1\right)$, we first derive pricing functions that have Property~\ref{prop:stableCR}, and then prove the optimality of the functions. In the following derivation of the pricing functions, we assume that all pricing functions are continuous and non-decreasing, for the solution existence of our differential equations. However, the assumptions are not required by the proof of optimality. The following claim will be useful for the derivation.
\begin{claim}\label{clm:constFcn2}
	For any $\beta>-1$, if a pricing function $P\left(\rho\right)$ leads to a finite worst-case competitive ratio of $\alpha$, then $P\left(\rho\right)\allowbreak=\underline{p}, \forall\rho\in\left[0,1/\alpha\right]$.
\end{claim}
\begin{proof}
	For $\beta>0$, the proof is similar to that of Claim~\ref{clm:constFcn1} and is omitted. For $\beta\in\left(-1,0\right]$, the claim follows immediately from Theorem~\ref{thm:betaLeq_1}.
\end{proof}

Our derivation of the pricing function is further divided into two cases.

\vspace{2mm}
\noindent {\bf Case 1:} $\beta\in\left(\beta_{0},1\right)$ where $\beta_{0}\in\left(0,1\right)$, 
 such that $\beta>1/\alpha_{2}$
 and $\alpha_{2}$ is the worst-case competitive ratio achieved using the optimal pricing function for $\beta\in\left(\beta_{0},1\right)$. 
 According to Claim~\ref{clm:constFcn2}, the pricing function $P_{2}\left(\rho\right)=\underline{p},\forall\rho\in\left[0,1/\alpha_{2}\right]$. When $\rho\mbox{*}\in\left(1/\alpha_{2},1\right)$, as discussed for Eq.~\eqref{eq:V_ol1}, the minimum total value of an online solution is
\begin{equation}
V_{ol}\left(\rho\mbox{*}\right)=\int_{0}^{\rho\mbox{*}}P_{2}\left(\rho\right)d\rho,\label{eq:V_ol2}
\end{equation}
and hence
\begin{equation}
\frac{dV_{ol}\left(\rho\mbox{*}\right)}{d\rho\mbox{*}}=\frac{d\left(\int_{0}^{\rho\mbox{*}}P_{2}\left(\rho\right)d\rho\right)}{d\rho\mbox{*}}=P_{2}\left(\rho\mbox{*}\right),\label{eq:dVol/dRho}
\end{equation}
which is illustrated by the light blue area in Fig.~\ref{fig:priFcn2_auc}.
Since $P_{2}\left(\rho\right)$ is non-decreasing, when $\rho\mbox{*}\in\left(1/\alpha_{2},\beta\right]$, we still have $	V_{opt}\left(\rho\mbox{*}\right)=P_{2}\left(\rho\mbox{*}\right)-\epsilon$ as discussed for Eq.~\eqref{eq:V_opt1}, where $\epsilon$ is an arbitrarily small positive value. Thus
\begin{equation}
\frac{d\sup_{\epsilon>0}V_{opt}\left(\rho\mbox{*}\right)}{d\rho\mbox{*}}=\frac{dP_{2}\left(\rho\mbox{*}\right)}{d\rho\mbox{*}}.\label{eq:dVopt/dRho-rhoLeqBeta}
\end{equation}
It follows from Eq. \eqref{eq:stableCR_diff}, \eqref{eq:dVol/dRho} and \eqref{eq:dVopt/dRho-rhoLeqBeta} that
\begin{equation}
\frac{dP_{2}\left(\rho\right)}{d\rho}-\alpha_{2}P_{2}\left(\rho\right)=0.\label{eq:diffEq-exp}
\end{equation}
Solving the differential equation above gives $P_{2}\left(\rho\right)=Ce^{\alpha_{2}\rho}$, where $C$ is a constant to be determined. Since we assumed the continuity of $P_{2}\left(\rho\right)$, we let $\lim_{\rho\rightarrow1/\alpha_{2}+}P_{2}\left(\rho\right)=P_{2}\left(1/\alpha_{2}\right)=\underline{p}$, and then we obtain $C=\underline{p}/e$, and $P_{2}\left(\rho\right)=\underline{p}e^{\alpha_{2}\rho-1},\forall\rho\in\left(1/\alpha_{2},\beta\right]$.

When $\rho\mbox{*}\in\left(\beta,1\right)$, having a set of users with a unit value of $P_{2}\left(\rho\mbox{*}\right)-\epsilon$ to consume 
 all resource is no longer possible in the worst case. Instead, there can be a set of unsatisfied users with a total demand of $1+\beta-\rho\mbox{*}$, and with a unit value of $P_{2}\left(\rho\mbox{*}\right)-\epsilon$, such that the optimal offline solution yields the maximum optimal offline total value given Eq.~\eqref{eq:V_ol2}:
\begin{equation}
V_{opt}\left(\rho\mbox{*}\right)=\left(1+\beta-\rho\mbox{*}\right)\left(P_{2}\left(\rho\mbox{*}\right)-\epsilon\right)+\int_{\beta}^{\rho\mbox{*}}P_{2}\left(\rho\right)d\rho,\label{eq:Vopt-rhoGtBeta}
\end{equation}
as shown by the red and yellow AUCs (partially covered by the blue one) in Fig.~\ref{fig:priFcn2_auc}. We have
\begin{equation}
\frac{d\sup_{\epsilon>0}V_{opt}\left(\rho\mbox{*}\right)}{d\rho\mbox{*}}=\left(1+\beta-\rho\mbox{*}\right)\frac{dP_{2}\left(\rho\mbox{*}\right)}{d\rho\mbox{*}},\label{eq:dVopt/dRho-rhoGtBeta}
\end{equation}
which is illustrated by the light red areas in Fig.~\ref{fig:priFcn2_auc}. Note that, there can be a case which leads to a larger optimal offline total value, by increasing the value corresponding to $\rho\in\left[\beta,\rho\mbox{*}\right]$ (\ie, the yellow AUC in Fig.~\ref{fig:priFcn2_auc}). Suppose the increased optimal offline total value is $V_{opt}\left(\rho\mbox{*}\right)+\Delta$ ($\Delta>0$), the online total value will also be increased to $V_{ol}\left(\rho\mbox{*}\right)+\Delta$. However, since the competitive ratio now changes to $\sup_{\epsilon>0}\frac{V_{opt}\left(\rho\mbox{*}\right)+\Delta}{V_{ol}\left(\rho\mbox{*}\right)+\Delta}<\sup_{\epsilon>0}\frac{V_{opt}\left(\rho\mbox{*}\right)}{V_{ol}\left(\rho\mbox{*}\right)}$, it cannot be the worst case.

It follows from Eq. \eqref{eq:stableCR_diff}, \eqref{eq:dVol/dRho} and \eqref{eq:dVopt/dRho-rhoGtBeta} that
\begin{equation}
\left(1+\beta-\rho\right)\frac{dP_{2}\left(\rho\right)}{d\rho}-\alpha_{2}P_{2}\left(\rho\right)=0.\label{eq:diffEq-pow}
\end{equation}
Solving the differential equation above gives $P_{2}\left(\rho\right)=C(1+\beta-\rho)^{-\alpha_{2}}$, where $C$ is a constant to be determined. Again, due to the continuity of $P_{2}\left(\rho\right)$, we let $\lim_{\rho\rightarrow\beta+}\allowbreak P_{2}\left(\rho\right)=P_{2}\left(\beta\right)=\underline{p}e^{\alpha_{2}\beta-1}$. Then we obtain $C=\underline{p}e^{\alpha_{2}\beta-1}$, and $P_{2}\left(\rho\right)=\underline{p}e^{\alpha_{2}\beta-1}\allowbreak\left(1+\beta-\rho\right)^{-\alpha_{2}},\forall\rho\in\left(\beta,1\right]$. To have a constant competitive ratio at $\rho\mbox{*}=1-$ and $\rho\mbox{*}=1$, as suggested by Property \ref{prop:stableCR}, we let $P_{2}\left(1\right)=\underline{p}e^{\alpha_{2}\beta-1}\beta^{-\alpha_{2}}=\overline{p}=\gamma\underline{p}$, which leads to
\begin{equation}
\alpha_{2}=\frac{\log\gamma+1}{\beta-\log\beta}.\label{eq:alpha_2}
\end{equation}

To obtain the value of $\beta_{0}$, let $\beta=\beta_{0}=1/\alpha_{2}$. By Eq. \eqref{eq:alpha_2}, we obtain
\begin{equation}
	\beta_{0}=\frac{W\left(\log\gamma\right)}{\log\gamma}.\label{eq:beta_0}
\end{equation}
Here, $W\left(\cdot\right)$ is the Lamber $W$-function ({\em a.k.a.}~the omega function or the product logarithm), which is the inverse function of $f\left(W\right)=We^{W}$. Therefore, for $\beta\in\left(\beta_{0},1\right)$, the pricing function is
\begin{equation}
	P_{2}\left(\rho\right)=\begin{cases}
	\underline{p}, & \rho\in\left[0,1/\alpha_{2}\right]\\
	\underline{p}e^{\alpha_{2}\rho-1}, & \rho\in\left(1/\alpha_{2},\beta\right]\\
	\underline{p}e^{\alpha_{2}\beta-1}\left(1+\beta-\rho\right)^{-\alpha_{2}}, & \rho\in\left(\beta,1\right)\\
	+\infty, & \rho=1
	\end{cases}.\label{eq:priFcn2}
\end{equation}
An example of $P_{2}\left(\rho\right)$ is shown in Fig.~\ref{fig:priFcn2_curv} by the dashed line corresponding to $\beta=0.5$, where $\beta_{0}=0.399$. In practice, $\beta$ can be estimated or optimized against competitive ratios.

\begin{figure}[!t]
	\begin{centering}
		\subfloat[The worst-case $V_{ol}(\rho\mbox{*})$ and $V_{opt}(\rho\mbox{*})$
		for $\beta=0.5$, $\rho\mbox{*}=0.7$ visualized by AUCs.]{
			\includegraphics[width=0.45\textwidth]{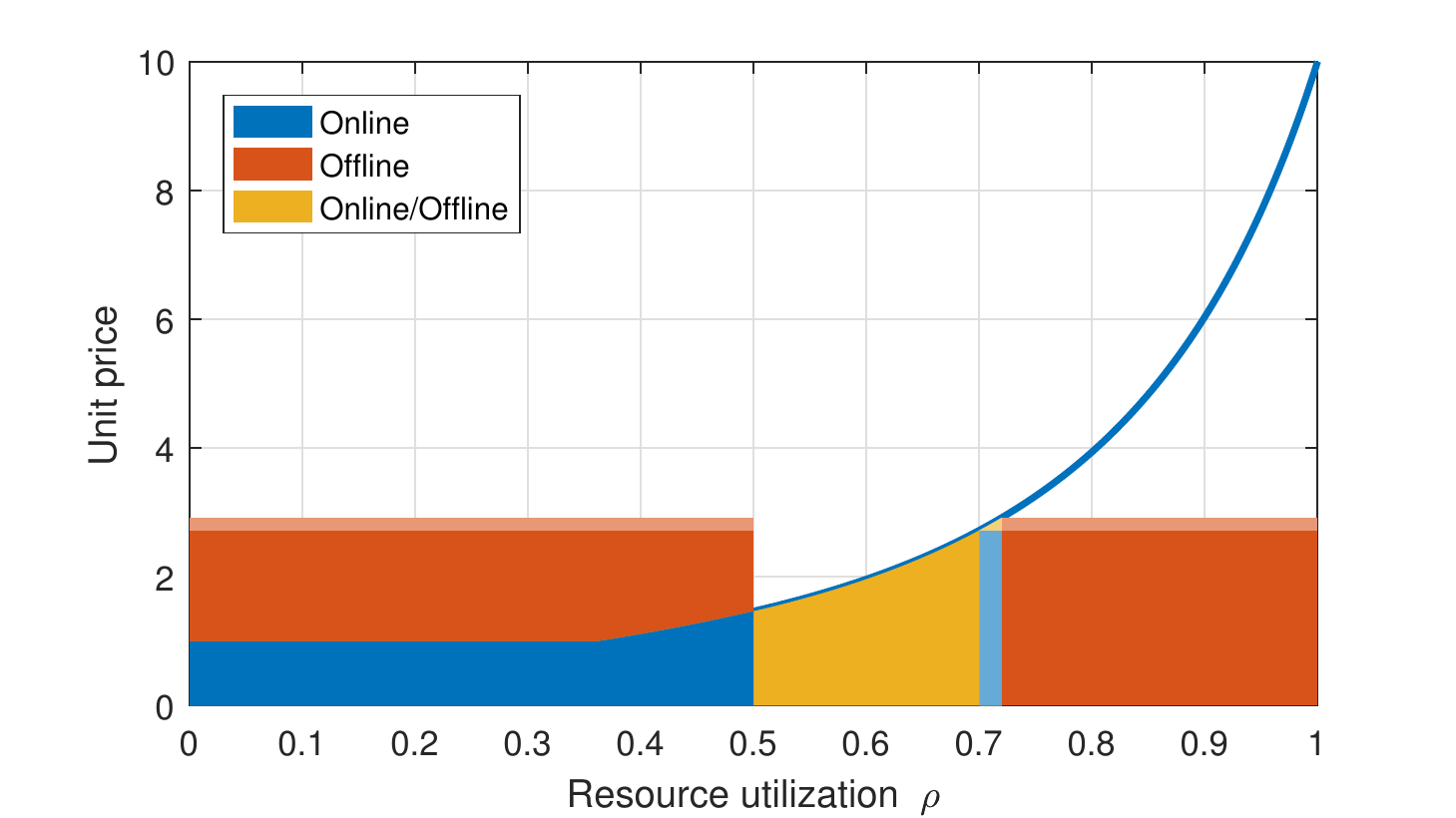}
			\label{fig:priFcn2_auc}}
		\par\end{centering}
	\vspace*{-10pt}
	\begin{centering}
		\subfloat[Pricing functions and competitive ratios for different values of $\beta$.]{
			\includegraphics[width=0.45\textwidth]{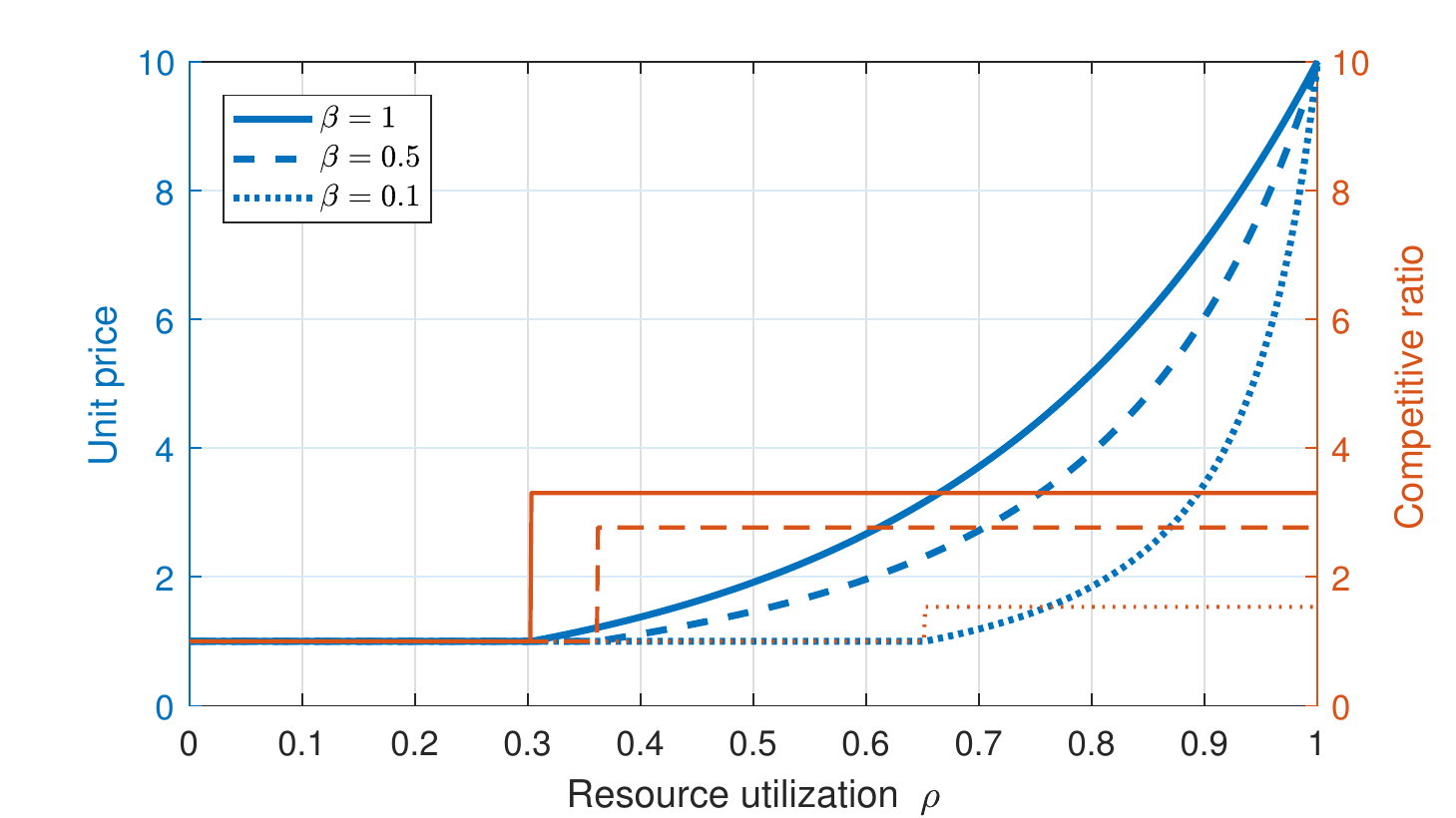}
			\label{fig:priFcn2_curv}}
		\par\end{centering}
	\vspace*{0pt}
	\caption{Pricing functions and competitive ratios for 
	$\beta\in (0,1)$, $\underline{p}=1$, $\overline{p}=10$.}
	\label{fig:priFcn2}
\end{figure}

\textcolor{black}{
\begin{theorem}\label{thm:betaGeq_b0}
	For $\beta\in\left(\beta_{0},1\right)$, the pricing function in \eqref{eq:priFcn2} is optimal according to Definition~\ref{def:optFcn}, and the corresponding worst-case competitive ratio is $\alpha_{2}$.
\end{theorem}}
\begin{proof}
	The proof of the worst-case competitive ratio $\alpha_{2}$ follows that of Theorem~\ref{thm:priFcn1_cr}, and is omitted.
	
	Suppose there exists a pricing function, $P'_{2}\left(\rho\right)$, that achieves a worst-case competitive ratio $\alpha'_{2}<\alpha_{2}$. According to Claim~\ref{clm:constFcn2} and the proof of Theorem~\ref{thm:priFcn1_opt}, we have
	\[
	\int_{0}^{\beta}P'_{2}\left(\rho\right)d\rho<\int_{0}^{\beta}P_{2}\left(\rho\right)d\rho,
	\]
	where $P_{2}\left(\rho\right)$ is the pricing function in \eqref{eq:priFcn2}.
	
	If there exists some $\rho\in\left(\beta,1\right)$ such that $P'_{2}\left(\rho\right)\geq P_{2}\left(\rho\right)$ we find the smallest one, and denote it by $\rho_{1}$. Then there can be a case where $\rho\mbox{*}=\rho_{1}$, such that the online total value
	\[
	V'_{ol}\left(\rho\mbox{*}\right)=\int_{0}^{\rho_{1}}P'_{2}\left(\rho\right)d\rho<\int_{0}^{\rho_{1}}P_{2}\left(\rho\right)d\rho-\Delta=V_{ol}\left(\rho\mbox{*}\right)-\Delta,
	\]
	where $\Delta=\int_{\beta}^{\rho_{1}}\left[P_{2}\left(\rho\right)-P'_{2}\left(\rho\right)\right]d\rho$; while the optimal offline total value $V'_{opt}\left(\rho\mbox{*}\right)\geq V_{opt}\left(\rho\mbox{*}\right)-\Delta$ according to Eq.~\eqref{eq:Vopt-rhoGtBeta}. Thus the worst-case competitive ratio $\alpha'_{2}\geq\sup_{\epsilon>0}\frac{V'_{opt}\left(\rho\mbox{*}\right)}{V'_{ol}\left(\rho\mbox{*}\right)}>\sup_{\epsilon>0}\allowbreak\frac{V_{opt}\left(\rho\mbox{*}\right)-\Delta}{V_{ol}\left(\rho\mbox{*}\right)-\Delta}>\alpha_{2}$, contradicting the assumption $\alpha'_{2}<\alpha_{2}$. Therefore, $P'_{2}\left(\rho\right)<P_{2}\left(\rho\right),\forall\rho\in\left(\beta,1\right)$.
	
	For $\rho\mbox{*}=1$, since $P'_{2}\left(1\right)\leq\overline{p}$ (a unit price higher than $\overline{p}$ will reject all potential users) is finite, we now have
	\[
	V'_{ol}\left(\rho\mbox{*}\right)=\int_{0}^{1}P'_{2}\left(\rho\right)d\rho<\int_{0}^{1}P_{2}\left(\rho\right)d\rho-\Delta=V_{ol}\left(\rho\mbox{*}\right)-\Delta,
	\]
	where $\Delta=\int_{\beta}^{1}\left[P_{2}\left(\rho\right)-P'_{2}\left(\rho\right)\right]d\rho$. However, as the resource is exhausted, subsequent users will not be satisfied regardless of their valuations. There can be a case where the optimal offline total value $V'_{opt}\left(\rho\mbox{*}\right)=V_{opt}\left(\rho\mbox{*}\right)-\Delta$ according to Eq.~\eqref{eq:Vopt-rhoGtBeta}. Thus the worst-case competitive ratio $\alpha'_{2}\geq\frac{V'_{opt}\left(\rho\mbox{*}\right)}{V'_{ol}\left(\rho\mbox{*}\right)}>\frac{V_{opt}\left(\rho\mbox{*}\right)-\Delta}{V_{ol}\left(\rho\mbox{*}\right)-\Delta}>\alpha_{2}$, contradicting the assumption $\alpha'_{2}<\alpha_{2}$.
\end{proof}

\noindent {\bf Case 2:} $\beta\in\left(0,\beta_{0}\right]$. From the definition of $\beta_{0}$, we have $\beta\leq1/\alpha_{3}$, where $\alpha_{3}$ is the worst-case competitive ratio of the optimal pricing function in this case. According to Claim~\ref{clm:constFcn2}, the pricing function $P_{3}\left(\rho\right)=\underline{p},\forall\rho\in\left[0,1/\alpha_{3}\right]$. When $\rho\mbox{*}\in\left(1/\alpha_{3},1\right)$, $V_{ol}\left(\rho\mbox{*}\right)$ follows Eq.~\eqref{eq:dVol/dRho} with $P_{2}\left(\rho\mbox{*}\right)$ replaced by $P_{3}\left(\rho\mbox{*}\right)$; $V_{opt}\left(\rho\mbox{*}\right)$ follows Eq.~\eqref{eq:Vopt-rhoGtBeta}, \eqref{eq:dVopt/dRho-rhoGtBeta} with $P_{2}\left(\rho\mbox{*}\right)$ replaced by $P_{3}\left(\rho\mbox{*}\right)$. Then, following Eq.~\eqref{eq:diffEq-pow}, we have $P_{3}\left(\rho\right)=C(1+\beta-\rho)^{-\alpha_{3}}$. As discussed for Eq.~\eqref{eq:alpha_2}, we let $\lim_{\rho\rightarrow1/\alpha_{3}+}\allowbreak P_{3}\left(\rho\right)=P_{3}\left(1/\alpha_{3}\right)=\underline{p}$, $P_{3}\left(1\right)=\overline{p}=\gamma\underline{p}$. Solving the resulting equations:
\begin{eqnarray*}
	C\left(1+\beta-\frac{1}{\alpha_{3}}\right)^{-\alpha_{3}} & = & \underline{p},\\
	C\beta^{-\alpha_{3}} & = & \gamma\underline{p},
\end{eqnarray*}
we get
\begin{equation}
	\alpha_{3}=\frac{\log\gamma}{\left(1+\beta\right)\log\gamma-W\left(\beta\gamma^{1+\beta}\log\gamma\right)},\label{eq:alpha_3}
\end{equation}
and the pricing function for $\beta\in\left(0,\beta_{0}\right]$ is:
\begin{equation}
	P_{3}\left(\rho\right)=\begin{cases}
	\underline{p}, & \rho\in\left[0,1/\alpha_{3}\right]\\
	\underline{p}\gamma\beta^{\alpha_{3}}\left(1+\beta-\rho\right)^{-\alpha_{3}}, & \rho\in\left(1/\alpha_{3},1\right)\\
	+\infty, & \rho=1
	\end{cases}.\label{eq:priFcn3}
\end{equation}
An example of $P_{3}\left(\rho\right)$ is shown in Fig.~\ref{fig:priFcn2_curv} by the dashed line corresponding to $\beta=0.2$.
\textcolor{black}{
\begin{theorem}\label{thm:betaGt_0}
	For $\beta\in\left(0,\beta_{0}\right]$, the pricing function in \eqref{eq:priFcn3} is optimal according to Definition~\ref{def:optFcn}, and the corresponding worst-case competitive ratio is $\alpha_{3}$.
\end{theorem}}
\begin{proof}
	The proof is similar to that of Theorem~\ref{thm:betaGeq_b0} and is omitted.
\end{proof}

To provide a better understanding on how $\beta\in(0,1)$ affects the competitive ratio as dictated by Theorems~\ref{thm:betaGeq_1}, \ref{thm:betaGeq_b0} and \ref{thm:betaGt_0}, we plot the competitive ratio as a function of $\beta$ in Fig.~\ref{fig:priFcn2_cr}. \textcolor{black}{As shown in the figure, for a certain value of $\gamma$, the competitive ratio decreases with the decrease of $\beta$, and reaches the minimum value $1$ when $\beta$ drops to $0$}.

\begin{figure}[!t]
	\begin{centering}
		\includegraphics[width=0.45\textwidth]{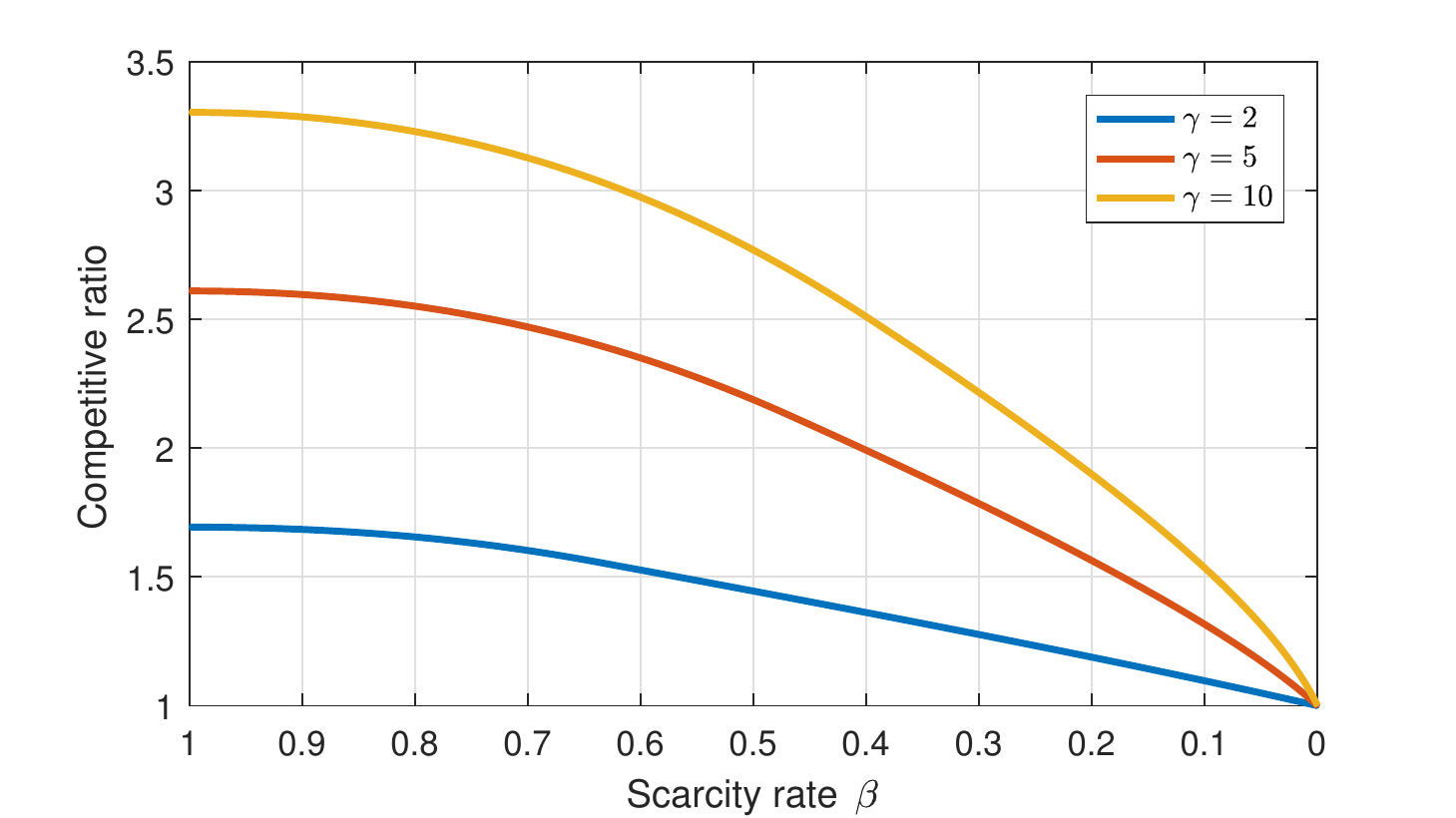}
		\par\end{centering}
	\protect\caption{Competitive ratios for different values of $\beta$ and $\gamma$.}
	\label{fig:priFcn2_cr}
\end{figure}

\vspace{2mm}
Putting Eq.~\eqref{eq:priFcn1}, \eqref{eq:priFcn2} and \eqref{eq:priFcn3} together, we have obtained a 2-dimensional piecewise pricing function, $P\left(\rho;\beta\right)$. An illustration of the pricing function is given in Fig.~\ref{fig:priFcn2_surf}.
\begin{figure}[!t]
	\begin{centering}
		\includegraphics[width=0.45\textwidth]{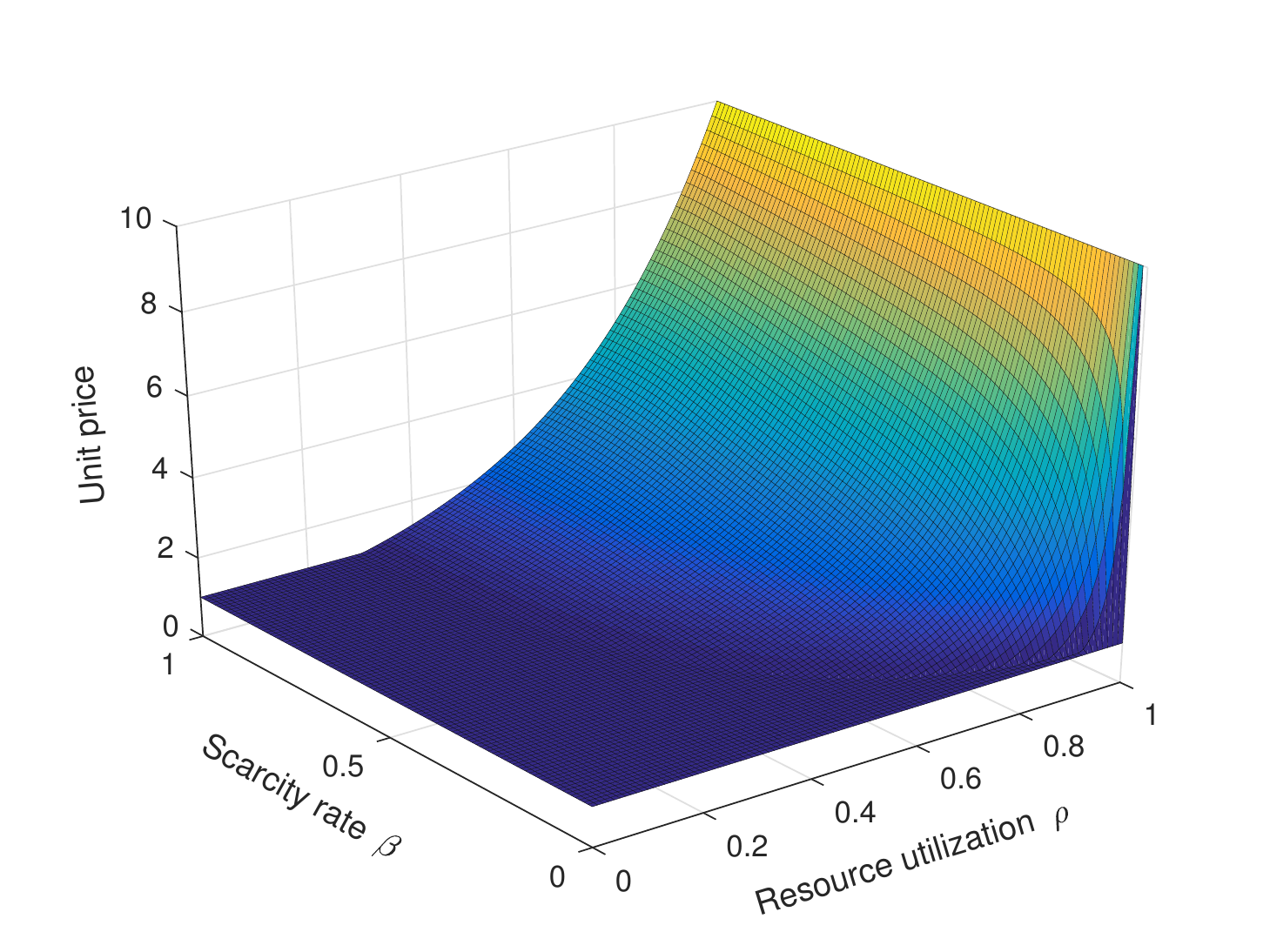}
		\par\end{centering}
	\protect\caption{2-D pricing function $P\left(\rho;\beta\right)$ for $\rho\in\left[0,1\right],\beta\in\left[0,1\right]$.}
	\label{fig:priFcn2_surf}
\end{figure}

\subsection{Linear Operational Cost}
\label{sec:linOpCost}
Resource provisioning in real-world cloud computing systems often incurs an operational cost. If such cost is proportional to the amount of resources provisioned, then we have a linear operational cost \cite{zhang2015online}. We can extend the proposed pricing strategy to accommodate such linear operational cost by making two modifications. First, we replace Assumption~\ref{asmp:varCstr} by:
\begin{assumption}\label{asmp:varCstr1}
	The variability of users' valuations is constrained, \ie, $\underline{p}+c\leq v_{i}/d_{i}\leq\overline{p}+c$.
\end{assumption}
Here, $c\geq 0$ is the operational cost of using a unit of resource. Second, we replace the pricing functions \eqref{eq:priFcn1}, \eqref{eq:priFcn2} and \eqref{eq:priFcn3}, by $P'_{1}\left(\rho\right)=P_{1}\left(\rho\right)+c$, $P'_{2}\left(\rho\right)=P_{2}\left(\rho\right)+c$ and $P'_{3}\left(\rho\right)=P_{3}\left(\rho\right)+c$. Then it is clear that all discussions about the proposed pricing strategy remain valid, including the proof of optimality. In the rest of this paper, we ignore operational cost  for simplicity.

%% file: genPricStrat.tex
\section{Pricing Multiple Resource Types with Resource Recycling}
\label{sec:pricing-general}

In this section, we extend our elementary resource allocation problem in \eqref{eq:prob1} to one with multiple types of resources (Sec.~\ref{sec:multi_resource_type}), and then further investigate the practical case that resource usage of a user lasts for multiple time slots (Sec.~\ref{sec:multi_timeslots}). \textcolor{black}{We show that, by carefully designing the pricing and scheduling strategy, the worst-case competitive ratio in social welfare will not be influenced by the number of resource types, or by the number of requested time slots.}

\subsection{Pricing Function for Multiple Types of Resources}
\label{sec:multi_resource_type}

Now we consider a cloud system that provides multiple types of resources in a set $\mathcal{R}$, as exemplified by CPU, GPU, RAM, and disk storage. Let $d_{i,r}$ be user $i$'s demand for resource $r$, $\forall r\in \mathcal{R}$. Again, we assume the total amount of each type of resource is 1, so that $d_{i,r}$ is the proportion of the overall supply of resource $r$ demanded by $i$. 

The offline social welfare maximization problem is:
	\begin{equation}\label{eq:prob2}
	\textrm{maximize}\quad \sum_{i\in\mathcal{U}}v_{i}x_{i}
	\end{equation}
	
	\vspace*{-3mm}
	\quad\text{s.t.:}
		\vspace*{-6mm}
	
	\begin{align*}
	\sum_{i\in\mathcal{U}}d_{i,r}x_{i}\leq 1, & \forall r\in\mathcal{R} \quad\quad\quad\quad\quad (\ref{eq:prob2}a)\\
	x_{i}\in\left\{0,1\right\}, & \forall i\in\mathcal{U} \quad\quad\quad\quad\quad (\ref{eq:prob2}b)
	\end{align*}
The online resource allocation algorithm we apply to determine $x_{i}$ immediately after user $i$ comes to the system, is the same as Alg.~\ref{alg:postedPrice}, 
except that $d_{i}$ and the pricing function will be redefined.

Given the optimal pricing functions \eqref{eq:priFcn1} (for $\beta\ge 1$), \eqref{eq:priFcn2} (for $\beta\in (\beta_0,1)$), \eqref{eq:priFcn3} (for $\beta\in (0, \beta_0]$) and \eqref{eq:priFcn4} (for $\beta\in (-1,0]$) in case of a single resource type, we can simply price each type of resource independently as $p_r$, using these pricing functions, and sum them up by $\sum_{r\in\mathcal{R}}d_{i,r}p_r$ to form a total price (a user $i$ accepts the prices and rents resources at quantities $d_{i,r}$'s, if and only if $v_i$ is no smaller than the total price). Before doing so, we need to redefine $\underline{p}$ and $\overline{p}$. One way is to define them for each type of resource independently, as $\underline{p}_{r}=\inf_{i}\frac{v_{i}}{d_{i,r}}$, and $\overline{p}_{r}=\sup_{i}\frac{v_{i}}{d_{i,r}}$, as done by \citeauthor{zhang2015online} \cite{zhang2015online}. However, a drawback of this definition is that $\overline{p}_{r}$ can be infinite, as we do not assume that every user demands all types of resources. A remedy to this problem is to define the same $\underline{p}$ and $\overline{p}$ for all types of resources, as $\underline{p}=\inf_{i}\frac{v_{i}}{d_{i}}$, and $\overline{p}=\sup_{i}\frac{v_{i}}{d_{i}}$, where $d_{i}=\sum_{r\in\mathcal{R}}d_{i,r}$. In this way, Assumption~\ref{asmp:varCstr} or \ref{asmp:varCstr1} remains intact. \textcolor{black}{The definitions of $\underline{p}$ and $\overline{p}$ are a simple extension of Assumption \ref{asmp:varCstr} for the multi-resource case. Compared to the former definition, they do not make any (implicit) assumptions on the ratio of different resources each user demands, and thus are more practical. Moreover, summing up the demand for different types of resources is reasonable when each $d_{i,r}$ is normalized by the total supply of the corresponding resource, such that their values are all in the range of $[0,1]$.} Then given the resource utilization $\rho_{r}$ and scarcity level $\beta_{r}$ of each type of resource $r\in\mathcal{R}$, we define an average unit price for any resource for user $i$ as
\begin{equation}
\mathscr{P}_{i}\left(\boldsymbol{\rho}\right)=\frac{1}{d_{i}}\sum_{r\in\mathcal{R}}d_{i,r}P\left(\rho_{r};\beta_{r}\right)\label{eq:priFcn-multiRes}
\end{equation}
where $\boldsymbol{\rho}$ denotes the vector of $\rho_{r},\forall r\in\mathcal{R}$, and $P\left(\rho_{r};\beta_{r}\right)$ is defined by Eq.~\eqref{eq:priFcn1}, \eqref{eq:priFcn2}, \eqref{eq:priFcn3} and \eqref{eq:priFcn4}. Therefore, $d_{i}\mathscr{P}_{i}\left(\boldsymbol{\rho}\right)$ is the total price for user $i$. Note that we omit $\beta_{r}$ in $\mathscr{P}_{i}\left(\boldsymbol{\rho}\right)$ for notational simplicity, but different $\beta_{r}$ will lead to different $\mathscr{P}_{i}\left(\boldsymbol{\rho}\right)$.

While it is quite straightforward to adapt the pricing strategy for a single resource type to the case of multiple resource types, the resulting worst-case competitive ratio of social welfare will be different. Specifically, we denote the final resource utilization level of resource $r$ by $\rho_{r}\mbox{*}, \forall r\in\mathcal{R}$, according to Definition~\ref{def:finRho}, and we analyze competitive ratios in three cases: (i) $\rho_{r}\mbox{*}\in\left[0,1/\alpha_{r}\right],\forall r\in\mathcal{R}$; (ii) there exists an $r\in\mathcal{R}$ such that $\rho_{r}\mbox{*}\in\left(1/\alpha_{r},1\right)$, but no $r\in\mathcal{R}$ such that $\rho_{r}\mbox{*}=1$; (iii) there exists an $r\in\mathcal{R}$ such that $\rho_{r}\mbox{*}=1$. Here, $\alpha_{r}$ is defined by Eq.~\eqref{eq:alpha_1}, \eqref{eq:alpha_2} or \eqref{eq:alpha_3} for $\beta=\beta_{r}$. We denote the three cases by $\boldsymbol{\rho}\mbox{*}\in\Omega_{1}$, $\boldsymbol{\rho}\mbox{*}\in\Omega_{2}$ and $\boldsymbol{\rho}\mbox{*}\in\Omega_{3}$, respectively, and observe that $\Omega_{1}\cup\Omega_{2}\cup\Omega_{3}$ covers all possible values of $\boldsymbol{\rho}\mbox{*}$.
Without loss of generality, here we assume not all $\beta_{r}\leq 0$, since otherwise the worst-case competitive ratio would be $1$.

\begin{lemma}\label{lem:Omega_1_CR}
	For $\boldsymbol{\rho}\mbox{*}\in\Omega_{1}$, the worst-case competitive ratio achieved by Alg.~\ref{alg:postedPrice} using pricing function (\ref{eq:priFcn-multiRes}) for multiple types of resources is $\alpha^{1}=1$.
\end{lemma}
\begin{proof}
	For $\boldsymbol{\rho}\mbox{*}\in\Omega_{1}$, according to the pricing function in \eqref{eq:priFcn-multiRes}, $\mathscr{P}_{i}\left(\boldsymbol{\rho}\mbox{*}\right)=\underline{p}$, which by Assumption~\ref{asmp:varCstr} is acceptable to any potential users, thus the total demand of all users for resource $r$ is exactly $\rho_{r}\mbox{*}$. The social welfare achieved by the pricing function in \eqref{eq:priFcn-multiRes} is the total value of all users, which is also the maximum possible social welfare achieved by solving the offline problem \eqref{eq:prob2}. Therefore, the worst-case competitive ratio $\alpha^{1}=1$.
\end{proof}

For $\boldsymbol{\rho}\mbox{*}\in\Omega_{2}$, we first present the following claim, which states that worst cases happen when all users demand only one specific type of resource, driving the average unit price to go a bit over $\underline{p}$.
\begin{claim}\label{clm:Omega_2_wc}
	Let $\rho_{r}\mbox{*}\in\left[0,1/\alpha_{r}\right]$ for $r\in\mathcal{R}_{1}$, $\rho_{r}\mbox{*}\in\left(1/\alpha_{r},1\right)$ for $r\in\mathcal{R}_{2}$, where $\mathcal{R}_{1}\cup\mathcal{R}_{2}=\mathcal{R}$. For $\boldsymbol{\rho}\mbox{*}\in\Omega_{2}$, there exists a worst case that happens when $\rho_{r}\mbox{*}=0$ for $r\in\mathcal{R}_{1}$, and $\rho_{r}\mbox{*}=1/\alpha_{r}+\epsilon$ for $r\in\mathcal{R}_{2}$, where $\left|\mathcal{R}_{2}\right|=1$. Here, $\epsilon$ is an arbitrarily small number.
\end{claim}

The proof can be found in the appendix.

\begin{lemma}\label{lem:Omega_2_CR}
	For $\boldsymbol{\rho}\mbox{*}\in\Omega_{2}$, the corresponding worst-case competitive ratio $\alpha^{2}=\alpha_{\overline{r}}\sum_{r\in\mathcal{R}}\min\left\{ 1,1+\beta_{r}\right\}$, where $\alpha_{\overline{r}}$ is defined by Eq.~\eqref{eq:alpha_1}, \eqref{eq:alpha_2} or \eqref{eq:alpha_3} for $\beta=\beta_{\overline{r}}$, and $\overline{r}=\argmax_{r\in\mathcal{R}}\alpha_{r}$.
\end{lemma}
\begin{proof}
	According to Claim~\ref{clm:Omega_2_wc}, we let $\rho_{r}\mbox{*}=0$ for $r\in\mathcal{R}_{1}$, and $\rho_{r}\mbox{*}=1/\alpha_{r}+\epsilon$ for $r\in\mathcal{R}_{2}$, and let $\left|\mathcal{R}_{2}\right|=1$. Then from Eq.~\eqref{eq:V_ol-Omega_2} and \eqref{eq:V_opt-Omega_2}, $\mathcal{R}_{2}=\left\{\overline{r}\right\}$ maximizes $\alpha\left(\boldsymbol{\rho}\mbox{*}\right)$, and thus is a worst case for $\boldsymbol{\rho}\mbox{*}\in\Omega_{2}$. The corresponding competitive ratio
	\begin{equation}
	\alpha^{2}=\sup_{\epsilon>0}\frac{\sum_{r\in\mathcal{R}}\underline{p}\min\left\{ 1,1+\beta_{r}\right\} }{\underline{p}\rho_{\overline{r}}\mbox{*}}=\alpha_{\overline{r}}\sum_{r\in\mathcal{R}}\min\left\{ 1,1+\beta_{r}\right\}. \label{eq:Omega_2_CR}
	\end{equation}
\end{proof}

For $\boldsymbol{\rho}\mbox{*}\in\Omega_{3}$, the following claim states that worst cases happen when all users that are satisfied by an online solution, demand only one specific type of resource until it is exhausted.
\begin{claim}\label{clm:Omega_3_wc}
	Let $\rho_{r}\mbox{*}\in\left[0,1\right)$ for $r\in\mathcal{R}_{3}$, $\rho_{r}\mbox{*}=1$ for $r\in\mathcal{R}_{4}$, where $\mathcal{R}_{3}\cup\mathcal{R}_{4}=\mathcal{R}$. For $\boldsymbol{\rho}\mbox{*}\in\Omega_{3}$, there exists a worst case that happens when $\rho_{r}\mbox{*}=0$ for $r\in\mathcal{R}_{3}$, and $\rho_{r}\mbox{*}=1$ for $r\in\mathcal{R}_{4}$, where $\left|\mathcal{R}_{4}\right|=1$.
\end{claim}

The proof can be found in the appendix.

\begin{lemma}\label{lem:Omega_3_CR}
	For $\boldsymbol{\rho}\mbox{*}\in\Omega_{3}$, the corresponding worst-case competitive ratio $\alpha^{3}\geq\alpha_{\overline{r}}\sum_{r\in\mathcal{R}}\min\left\{ 1,1+\beta_{r}\right\}$.
\end{lemma}
\begin{proof}
	According to Claim~\ref{clm:Omega_3_wc}, we let $\rho_{r}\mbox{*}=0$ for $r\in\mathcal{R}_{3}$, and $\rho_{r}\mbox{*}=1$ for $r\in\mathcal{R}_{4}$, and let $\left|\mathcal{R}_{4}\right|=1$. Then from Eq.~\eqref{eq:V_ol-Omega_3} and \eqref{eq:V_opt-Omega_3}, we have the worst-cast competitive ratio for $\boldsymbol{\rho}\mbox{*}\in\Omega_{3}$
	\begin{equation}
	\alpha^{3}=\max_{r'\in\mathcal{R}'}\frac{\sum_{r\in\mathcal{R}\backslash\left\{ r'\right\} }\overline{p}\min\left\{ 1,1+\beta_{r}\right\} +\alpha_{r'}\int_{0}^{1}P\left(\rho;\beta_{r'}\right)d\rho}{\int_{0}^{1}P\left(\rho;\beta_{r'}\right)d\rho},\label{eq:Omega_3_CR}
	\end{equation}
	where $R'=\left\{r|\beta_{r}>0\right\}$. Since it is assumed that $\beta_{\overline{r}}>0$, we have $\min\left\{ 1,1+\beta_{\overline{r}}\right\}=1$, and hence
	\[
	\frac{\alpha_{\overline{r}}\int_{0}^{1}P\left(\rho;\beta_{\overline{r}}\right)d\rho}{\int_{0}^{1}P\left(\rho;\beta_{\overline{r}}\right)d\rho}=\frac{\overline{p}\min\left\{ 1,1+\beta_{\overline{r}}\right\} }{\overline{p}/\alpha_{\overline{r}}}.
	\]
	And since $\alpha_{\overline{r}}\int_{0}^{1}P\left(\rho;\beta_{\overline{r}}\right)d\rho\leq\overline{p}\min\left\{ 1,1+\beta_{\overline{r}}\right\}$, we have
	\begin{align*}
	\alpha^{3}\geq &~ \frac{\sum_{r\in\mathcal{R}\backslash\left\{ \overline{r}\right\} }\overline{p}\min\left\{ 1,1+\beta_{r}\right\} +\alpha_{\overline{r}}\int_{0}^{1}P\left(\rho;\beta_{\overline{r}}\right)d\rho}{\int_{0}^{1}P\left(\rho;\beta_{\overline{r}}\right)d\rho}\\
	\geq &~ \frac{\sum_{r\in\mathcal{R}}\overline{p}\min\left\{ 1,1+\beta_{r}\right\} }{\overline{p}/\alpha_{\overline{r}}}=\alpha_{\overline{r}}\sum_{r\in\mathcal{R}}\min\left\{ 1,1+\beta_{r}\right\} . 
	\end{align*}
\end{proof}

By Lemma~\ref{lem:Omega_1_CR}, \ref{lem:Omega_2_CR} and \ref{lem:Omega_3_CR}, we have the following theorem:
\begin{theorem}\label{thm:CR-multiRes}
	The worst-case competitive ratio achieved by Alg.~\ref{alg:postedPrice} using the pricing function in \eqref{eq:diffEq-pow} for multiple types of resources is given by Eq.~\eqref{eq:Omega_3_CR}.
\end{theorem}
As shown by Lemma~\ref{lem:Omega_3_CR}, the worst-case competitive ratio for multiple resource types increases roughly linearly with the number of resource types. 
 However, from Claim~\ref{clm:Omega_2_wc}, \ref{clm:Omega_3_wc}, and the analysis above, it is clear that the worst cases happen in very extreme scenarios, where all satisfied users demand only one type of resource, which is rather unrealistic in practical cloud computing systems.
 In fact, the supply of and the demand for resources in a cloud computing system are often balanced to some extent, since otherwise the supply would be adjusted to better meet the demand of users and to improve the system efficiency. Hence, we make the following realistic assumption:
\begin{assumption}\label{asmp:balancedMultiRes}
	All types of resources share a common scarcity level, $\beta_{\mathcal{R}}>0$, and hence a common $\alpha_{\mathcal{R}}$ as defined by Eq.~\eqref{eq:alpha_1}, \eqref{eq:alpha_2} or \eqref{eq:alpha_3} for $\beta=\beta_{\mathcal{R}}$; and the final utilization vector, $\boldsymbol{\rho}\mbox{*}$, follows
	\begin{equation}\label{eq:balancedMultiRes}
	\frac{\min_{r\in\mathcal{R}}\rho_{r}\mbox{*}}{\max_{r\in\mathcal{R}}\rho_{r}\mbox{*}}\geq\eta.
	\end{equation}
\end{assumption}

Assumption~\ref{asmp:balancedMultiRes} leads to an improved competitive ratio.
\begin{theorem}\label{thm:CR-balancedMultiRes}
	Under Assumption~\ref{asmp:balancedMultiRes}, the worst-case competitive ratio with the pricing function in \eqref{eq:diffEq-pow} is upper bounded by a constant with respect to $\left|\mathcal{R}\right|$.
\end{theorem}
\begin{proof}
	It is easy to prove that Claim~\ref{clm:Omega_2_wc} and \ref{clm:Omega_3_wc} are still valid under Assumption~\ref{asmp:balancedMultiRes}. For $\boldsymbol{\rho}\mbox{*}\in\Omega_{2}$, any worst case gives $V_{ol}\left(\boldsymbol{\rho}\mbox{*}\right)=\left[1+\left(\left|\mathcal{R}\right|-1\right)\eta\right]\underline{p}/\alpha_{\mathcal{R}}$ and $V_{opt}\left(\boldsymbol{\rho}\mbox{*}\right)=\left|\mathcal{R}\right|\underline{p}$, and hence the corresponding competitive ratio $\alpha^{2}=\frac{\left|\mathcal{R}\right|}{1+\left(\left|\mathcal{R}\right|-1\right)\eta}\alpha_{\mathcal{R}}$. Since $\eta\leq 1$, we have $\alpha^{2}\leq\alpha_{\mathcal{R}}/\eta$. For $\boldsymbol{\rho}\mbox{*}\in\Omega_{3}$, as $\epsilon\rightarrow0$, any worst case gives
	\[
	V_{ol}\left(\boldsymbol{\rho}\mbox{*}\right)=\int_{0}^{1}P\left(\rho;\beta_{\mathcal{R}}\right)d\rho+\left(\left|\mathcal{R}\right|-1\right)\int_{0}^{\eta}P\left(\rho;\beta_{\mathcal{R}}\right)d\rho,
	\]
	and
	\begin{align*}
	V_{opt}\left(\boldsymbol{\rho}\mbox{*}\right)= &~ \alpha_{\mathcal{R}}V_{ol}\left(\boldsymbol{\rho}\mbox{*}\right)\\
	+ &~ \left(\left|\mathcal{R}\right|-1\right)\left(1+\beta_{\mathcal{R}}-\eta\right)\left(\overline{p}-P\left(\eta;\beta_{\mathcal{R}}\right)\right),
	\end{align*}
	and hence the corresponding competitive ratio
	\[
	\alpha^{3}=\alpha_{\mathcal{R}}+\frac{\left(1+\beta_{\mathcal{R}}-\eta\right)\left(\overline{p}-P\left(\eta;\beta_{\mathcal{R}}\right)\right)}{\int_{0}^{1}P\left(\rho;\beta_{\mathcal{R}}\right)d\rho/\left(\left|\mathcal{R}\right|-1\right)+\int_{0}^{\eta}P\left(\rho;\beta_{\mathcal{R}}\right)d\rho}.
	\]
	Let 
	
	\vspace*{-8mm}
	\[\theta=\frac{\left(1+\beta_{\mathcal{R}}-\eta\right)\left(\overline{p}-P\left(\eta;\beta_{\mathcal{R}}\right)\right)}{\int_{0}^{\eta}P\left(\rho;\beta_{\mathcal{R}}\right)d\rho},
	\]
	we have $\alpha^{3}<\alpha_{\mathcal{R}}+\theta$. Therefore, the worst-cast competitive ratio under Assumption~\ref{asmp:balancedMultiRes} is upper bounded by $\max\{\alpha_{\mathcal{R}}/\eta,\allowbreak\alpha_{\mathcal{R}}+\theta\}$.
\end{proof}


Theorem~\ref{thm:CR-balancedMultiRes} justifies the use of the pricing function in \eqref{eq:priFcn-multiRes}, which is a direct extension of the optimal pricing functions for the single resource type case, but achieves a reasonably good (degraded by a constant factor w.r.t. $\left|\mathcal{R}\right|$) competitive ratio in scenarios with multiple resource types.

\subsection{Pricing Function for Multiple Time Slots}
\label{sec:multi_timeslots}

In real-world cloud systems, a user job runs over its specified resource bundle in the cloud, across one or more time slots. Once the job is completed, the resources that it occupies are then released back to the cloud pool. Therefore, cloud resources can be reused over time. Let $\mathcal{T}$ denote the set of all time slots that the system spans, and $\mathcal{T}_{i}$ be the set of time slots when user $i$ requires to use resources. $y_{i}\left(t\right)$ is an indication function as follows:
\begin{equation}\label{eq:yt}
y_{i}\left(t\right)=\begin{cases}
1, & t\in\mathcal{T}_{i}\\
0, & \text{otherwise}
\end{cases}.
\end{equation}
The offline social welfare maximization problem becomes:
	\begin{equation}\label{eq:prob3}
	\textrm{maximize}\quad \sum_{i\in\mathcal{U}}v_{i}x_{i}
	\end{equation}
	
	\vspace*{-3mm}
	\quad\text{s.t.:}
		\vspace*{-6mm}
		
	\begin{align*}
	\sum_{i\in\mathcal{U}}d_{i,r}x_{i}y_{i}\left(t\right)\leq1,&\forall r\in\mathcal{R},t\in\mathcal{T} ~~\quad (\ref{eq:prob3}a)\\
	x_{i}\in\left\{0,1\right\},&\forall i\in\mathcal{U} \quad\quad\quad\quad (\ref{eq:prob3}b)
	\end{align*}
 Since $y_{i}\left(t\right)$ is input (not a variable) in this optimization problem, problem \eqref{eq:prob3} is still an ILP.
The online resource allocation algorithm we apply to determine $x_i$ upon the arrival of user $i$ is still the same as Alg.~\ref{alg:postedPrice}, 
except that $d_{i}$ and the pricing function will be redefined, and $y_{i}\left(t\right)$ needs to be further determined.

In fact, problem \eqref{eq:prob2} and problem \eqref{eq:prob3} are equivalent if we consider resource $r$ in different time slots to be of different resource types. More specifically, let $d_{i,r\left(t\right)}=d_{i,r}y_{i}\left(t\right)$, where $r\left(t\right)\in\mathcal{R}\left(t\right)$, and $t\in\mathcal{T}$. Then problem \eqref{eq:prob3} will have exactly the same form as problem \eqref{eq:prob3}. Therefore, according to Lemma~\ref{lem:Omega_3_CR} and Theorem~\ref{thm:CR-multiRes}, the worst-case competitive ratio will increase roughly linearly with the number of time slots, $\left|\mathcal{T}\right|$, if no other assumptions are made. If the number of slots required by each user is upper bounded, then the worst-case competitive ratio will increase roughly linearly with the maximum number of slots required by each user, which is also undesirable. Intuitively, this issue is caused by the fact that, if one of the time slots required by a user is unavailable (e.g., no available resources), then the demand of the user cannot be satisfied as a whole, even if other required slots are all available.

To address the aforementioned problem, we propose a strategy that satisfies users' demand in an elastic manner. Specifically, assuming we are allowed to satisfy user $i$ with any $\left|\mathcal{T}_{i}\right|$ slots in a larger set of time slots, $\mathcal{T}'_{i}\supseteq\mathcal{T}_{i}$, we can significantly improve the competitive ratio by choosing $\left|\mathcal{T}_{i}\right|$ slots from $\mathcal{T}'_{i}$ that yield the lowest total price. Concretely, the corresponding online resource scheduling strategy is that, we try to satisfy each user $i$ with $\left|\mathcal{T}_{i}\right|$ time slots chosen from $\mathcal{T}'_{i}$, and $\left|\mathcal{T}'_{i}\right|=\lceil\lambda\left|\mathcal{T}_{i}\right|\rceil$, where $\lambda$ is a constant factor. Here $\mathcal{T}'_{i}$ can be interpreted as the allowed (loosened) time interval for completing the user's job.
The overall price to user $i$ is computed as the minimum possible total price of $\left|\mathcal{T}_{i}\right|$ time slots selected from $\mathcal{T}'_{i}$.

\textcolor{black}{From the user perspective, the price each user receives is determined upon its arrival in the system, and does not change afterwards. A user $i$ accepts the price and leases resource at quantities $d_{i,r}$'s in the chosen $\left|\mathcal{T}_{i}\right|$ time slots, if and only if $v_{i}$ is no smaller than the overall price. Once a user accepts the price, its job is guaranteed to be completed within $\lambda\left|\mathcal{T}_{i}\right|$. If the provider tells that a job cannot be completed within $\lambda\left|\mathcal{T}_{i}\right|$, the job will receive an infinitely high price according to the pricing function upon arrival ({\em i.e.}, the user will reject the price and the job will not be executed).}

In fact, similar non-consecutive execution schemes have been implemented on Amazon EC2 Spot Instance \cite{amazon_ec2si_interruptions}, and have been discussed in the literature \cite{zhou2016efficient}. Here, we further justify the use of non-consecutive execution schemes from a theoretical point of view.
 
Without loss of generality, we assume both $\mathcal{T}_{i}$ and $\mathcal{T}'_{i}$ are consecutive time slots; and if $\mathcal{T}_{i}=\left[\tau_{i},\tau_{i}+\left|\mathcal{T}_{i}\right|-1\right]$, we let $\mathcal{T}'_{i}=\left[\tau_{i},\tau_{i}+\left|\mathcal{T}'_{i}\right|-1\right]$.
To formulate the offline version of the modified social welfare maximization problem, we can add the following constraints to problem \eqref{eq:prob3}:
	\begin{align*}\label{eq:prob3a_rsrCap}
	\sum_{t\in\mathcal{T}'_{i}}y_{i}\left(t\right)=\left|\mathcal{T}_{i}\right|,&\forall i\in\mathcal{U} \quad\quad\quad\quad\quad\quad (\ref{eq:prob3}c)\\
	y_{i}\left(t\right)\in\left\{0,1\right\},&\forall i\in\mathcal{U},t\in\mathcal{T} \quad\quad\quad (\ref{eq:prob3}d)
	\end{align*}
Note that $y_{i}\left(t\right)$ now follows Eq.~(\ref{eq:prob3}c) and (\ref{eq:prob3}d), instead of Eq.~\eqref{eq:yt}, and $y_{i}\left(t\right)$ becomes a variable. Therefore, the new problem is no longer an ILP. 

 We reuse the notation $\mathscr{P}_{i}\left(\cdot\right)$ to denote the pricing function for user $i$; and we reuse the symbols, $d_{i}$ and $\boldsymbol{\rho}$, to taken into account different time slots, \ie, $d_{i}=\left|\mathcal{T}_i\right|\sum_{r\in\mathcal{R}}d_{i,r}$, and $\boldsymbol{\rho}$ denotes the vector of $\rho_{r}\left(t\right),\forall r\in\mathcal{R},t\in\mathcal{T}$. The definitions of $\underline{p}$ and $\overline{p}$ remain the same, \ie, $\underline{p}=\inf_{i}\frac{v_{i}}{d_{i}}$ and $\overline{p}=\sup_{i}\frac{v_{i}}{d_{i}}$. Then under Assumption~\ref{asmp:balancedMultiRes}, our pricing strategy for online resource allocation can be described by the following pricing function:
\begin{equation}
	\mathscr{P}_{i}\left(\boldsymbol{\rho}\right)=\frac{1}{d_{i}}\min_{\boldsymbol{y}_{i}\in\mathcal{Y}_{i}}\left[\sum_{t\in\mathcal{T}'_{i}}\sum_{r\in\mathcal{R}}d_{i,r}y_{i}\left(t\right)P\left(\rho_{r}\left(t\right);\beta_{\mathcal{R}}\right)\right],\label{eq:P-multiSlots}
\end{equation}
where $\mathcal{Y}_{i}$ is defined by Eq.~(\ref{eq:prob3}c) and (\ref{eq:prob3}d), 
 and $P\left(\rho_{r}(t);\beta_{\mathcal{R}}\right)$ is defined by Eq.~\eqref{eq:priFcn1}, \eqref{eq:priFcn2} and \eqref{eq:priFcn3}.
 Obviously, $d_{i}\mathscr{P}_{i}\left(\boldsymbol{\rho}\right)$ is the total price for user $i$.
 
 \textcolor{black}{In general, $\mathscr{P}_{i}\left(\boldsymbol{\rho}\right)$ sets different unit prices for different time slots, according to the scheduled resource utilization levels. Note that, the overall price that each user receives for its resource demand over the requested resource usage duration is determined when the user comes to the system and requests resources, and does not change over the course.}


Given an arbitrary set of time slots $\mathcal{T}$, and the corresponding time horizon $\left|\mathcal{T}\right|$, any $\mathcal{T}_{i}\nsubseteq\mathcal{T}$ can be ignored since it cannot be satisfied anyway. Furthermore, we ignore the marginal effect of any $\mathcal{T}'_{i}\nsubseteq\mathcal{T}$, since $\left|\mathcal{T}\right|$ is usually significantly larger than $\left|\mathcal{T}'_{i}\right|$. Thus, we assume $\mathcal{T}_{i},\mathcal{T}'_{i}\subseteq\mathcal{T},\forall i\in\mathcal{U}$. As we did to analyze competitive ratios for multiple resource types, we divide possible values of final resource utilization levels into three cases: (i) $\rho_{r}\mbox{*}\left(t\right)\in\left[0,1/\alpha_{R}\right],\forall r\in\mathcal{R},t\in\mathcal{T}$; (ii) there exists an $r\in\mathcal{R}$ and a $t\in\mathcal{T}$ such that $\rho_{r}\mbox{*}\left(t\right)\in\left(1/\alpha_{R},1\right)$, but no $r\in\mathcal{R}$ or $t\in\mathcal{T}$ such that $\rho_{r}\mbox{*}\left(t\right)=1$; (iii) there exists an $r\in\mathcal{R}$ and a $t\in\mathcal{T}$ such that $\rho_{r}\mbox{*}\left(t\right)=1$. Here, $\alpha_{R}$ is defined by Eq.~\eqref{eq:alpha_1}, \eqref{eq:alpha_2} or \eqref{eq:alpha_3} for $\beta=\beta_{R}$. We denote the three cases by $\boldsymbol{\rho}\mbox{*}\in\Pi_{1}$, $\boldsymbol{\rho}\mbox{*}\in\Pi_{2}$ and $\boldsymbol{\rho}\mbox{*}\in\Pi_{3}$, respectively.

\begin{lemma}\label{lem:Pi_1_CR}
	For $\boldsymbol{\rho}\mbox{*}\in\Pi_{1}$, the worst-case competitive ratio achieved by our online resource scheduling strategy using pricing function (\ref{eq:P-multiSlots}) is $\alpha^{1}=1$.
\end{lemma}
\begin{proof}
	The proof is similar to that of Lemma~\ref{lem:Omega_1_CR} and is omitted.
\end{proof}

\begin{lemma}\label{lem:Pi_2_CR}
	For $\boldsymbol{\rho}\mbox{*}\in\Pi_{2}$, the corresponding worst-case competitive ratio $\alpha^{2}<\frac{\alpha_{\mathcal{R}}}{\left(\lambda-1\right)\eta}+1$, where $\eta$ is defined as in Assumption~\ref{asmp:balancedMultiRes}.
\end{lemma}
\begin{proof}
	Let $\mathcal{T}_{1}=\left\{ t|\rho_{r}\mbox{*}\left(t\right)\in\left[0,1/\alpha_{\mathcal{R}}\right],\forall r\in\mathcal{R}\right\} $, and $\mathcal{T}_{2}=\{ t|\allowbreak\rho_{r}\mbox{*}\left(t\right)\in\left(1/\alpha_{\mathcal{R}},1\right),\exists r\in\mathcal{R}\} $, and $\mathcal{T}_{1}\cup\mathcal{T}_{2}=\mathcal{T}$. For $\boldsymbol{\rho}\mbox{*}\in\Pi_{2}$, following the proof of Lemma~\ref{clm:Omega_2_wc}, there exists a worst case that happens when $\rho_{r}\mbox{*}\left(t\right)=0$ for all $r\in\mathcal{R}$ and $t\in\mathcal{T}_{1}$; while for $t\in\mathcal{T}_{2}$, $\rho_{r}\mbox{*}\left(t\right)=1/\alpha_{r}+\epsilon$ for some $r'\in\mathcal{R}$, and $\rho_{r}\mbox{*}\left(t\right)=\eta\left(1/\alpha_{r}+\epsilon\right)$ for $r\in\mathcal{R}\backslash\left\{r'\right\}$. Here, $\epsilon$ is an arbitrarily small number. Following the proof of Theorem~\ref{thm:CR-balancedMultiRes}, as $\epsilon\rightarrow0$, we have
	\[
	V_{ol}\left(\boldsymbol{\rho}\mbox{*}\right)=\left|\mathcal{T}_{2}\right|\left[1+\left(\left|\mathcal{R}\right|-1\right)\eta\right]\underline{p}/\alpha_{\mathcal{R}},
	\]
	as the minimum total value of the online solution. For any $\mathcal{T}_{i}$, since $\left|\mathcal{T}'_{i}\right|=\lceil\lambda\left|\mathcal{T}_{i}\right|\rceil$, the demand will be satisfied regardless of the user's valuation, unless $\left|\mathcal{T}'_{i}\cap\mathcal{T}_{2}\right|>\lceil\frac{\lambda-1}{\lambda}\left|\mathcal{T}'_{i}\right|\rceil$. In other words, if the demand of user $i$ is not satisfied by the online solution, there must be at least $\lceil\frac{\lambda-1}{\lambda}\left|\mathcal{T}'_{i}\right|\rceil$ time slots in $\mathcal{T}'_{i}$ that also belong to $\mathcal{T}_{2}$; or equivalently, for any $\mathcal{S}\subseteq\mathcal{T}$, and any $\mathcal{T}'_{i}\subseteq\mathcal{S}$ that is not satisfied by the online solution, $\left|\mathcal{T}'_{i}\right|<\lfloor\frac{\lambda}{\lambda-1}\left|\mathcal{S}\cap\mathcal{T}_{2}\right|\rfloor$, and hence $\left|\mathcal{T}_{i}\right|<\lfloor\frac{1}{\lambda-1}\left|\mathcal{S}\cap\mathcal{T}_{2}\right|\rfloor$. Let $\mathcal{T}'_{2}$ be the union of all sets of consecutive time slots that contain $\mathcal{T}_{2}$, and have a cardinality of $\lfloor\frac{\lambda}{\lambda-1}\left|\mathcal{T}_{2}\right|\rfloor-1$. When $\left|\mathcal{T}_{i}\right|<\lfloor\frac{1}{\lambda-1}\left|\mathcal{S}\cap\mathcal{T}_{2}\right|\rfloor$, since at least one type of resource in at least one required time slot has a unit price above $\underline{p}$, there can be a set of users in a worst case, demanding all resources in $\left|\mathcal{T}'_{2}\right|$ time slots, with $	\mathscr{P}_{i}\left(\boldsymbol{\rho}\right)=\underline{p}$, where $\left|\mathcal{T}'_{2}\right|<2\lfloor\frac{1}{\lambda-1}\left|\mathcal{T}_{2}\right|\rfloor+\left|\mathcal{T}_{2}\right|$. Thus we have the maximum optimal offline total value
	\[
	V_{opt}\left(\boldsymbol{\rho}\mbox{*}\right)<\left|\mathcal{T}'_{2}\right|\left|\mathcal{R}\right|\underline{p}<\left(2\lfloor\frac{1}{\lambda-1}\left|\mathcal{T}_{2}\right|\rfloor+\left|\mathcal{T}_{2}\right|\right)\left|\mathcal{R}\right|\underline{p}.
	\]
	Therefore, for $\boldsymbol{\rho}\mbox{*}\in\Pi_{2}$, the worst-cast competitive ratio
	\begin{equation}\label{eq:Pi_2_CR}
	\begin{split}
	\alpha^{2}< &~ \frac{\frac{\lambda+1}{\lambda-1}\left|\mathcal{R}\right|\left|\mathcal{T}_{2}\right|\underline{p}}{\left|\mathcal{T}_{2}\right|\left[1+\left(\left|\mathcal{R}\right|-1\right)\eta\right]\underline{p}/\alpha_{\mathcal{R}}}\\
	= &~ \frac{\lambda+1}{\lambda-1}\frac{\left|\mathcal{R}\right|}{1+\left(\left|\mathcal{R}\right|-1\right)\eta}\alpha_{\mathcal{R}}\\
	\leq &~ \frac{\left(\lambda+1\right)\alpha_{\mathcal{R}}}{\left(\lambda-1\right)\eta}.
	\end{split}
	\end{equation}
\end{proof}

\begin{lemma}\label{lem:Pi_3_CR}
	For $\boldsymbol{\rho}\mbox{*}\in\Pi_{3}$, the corresponding worst-case competitive ratio $\alpha^{3}<\frac{\lambda+1}{\left(\lambda-1\right)\eta'}$, where $\eta'=\int_{0}^{\eta}P\left(\rho;\beta_{\mathcal{R}}\right)d\rho/\overline{p}$.
\end{lemma}
\begin{proof}
	Let $\mathcal{T}_{3}=\left\{ t|\rho_{r}\mbox{*}\left(t\right)\in\left[0,1\right),\forall r\in\mathcal{R}\right\} $, and $\mathcal{T}_{4}=\{ t|\rho_{r}\mbox{*}\left(t\right)\allowbreak=1,\exists r\in\mathcal{R}\} $, and $\mathcal{T}_{3}\cup\mathcal{T}_{4}=\mathcal{T}$. For $\boldsymbol{\rho}\mbox{*}\in\Pi_{2}$, following the proof of Lemma~\ref{clm:Omega_3_wc}, there exists a worst case that happens when $\rho_{r}\mbox{*}\left(t\right)=0$ for all $r\in\mathcal{R}$ and $t\in\mathcal{T}_{3}$; while for $t\in\mathcal{T}_{4}$, $\rho_{r}\mbox{*}\left(t\right)=1$ for some $r'\in\mathcal{R}$, and $\rho_{r}\mbox{*}\left(t\right)=\eta$ for $r\in\mathcal{R}\backslash\left\{r'\right\}$. Following the proof of Theorem~\ref{thm:CR-balancedMultiRes}, we have
	\[
	V_{ol}\left(\boldsymbol{\rho}\mbox{*}\right)=\left|\mathcal{T}_{4}\right|\left[\int_{0}^{1}P\left(\rho;\beta_{\mathcal{R}}\right)d\rho+\left(\left|\mathcal{R}\right|-1\right)\int_{0}^{\eta}P\left(\rho;\beta_{\mathcal{R}}\right)d\rho\right],
	\]
	as the minimum total value of the online solution. For any $\mathcal{T}_{i}$, since $\left|\mathcal{T}'_{i}\right|=\lceil\lambda\left|\mathcal{T}_{i}\right|\rceil$, the demand will be satisfied regardless of the user's valuation, unless $\left|\mathcal{T}'_{i}\cap\mathcal{T}_{4}\right|>\lceil\frac{\lambda-1}{\lambda}\left|\mathcal{T}'_{i}\right|\rceil$. In other words, if the demand of user $i$ is not satisfied by the online solution, there must be at least $\lceil\frac{\lambda-1}{\lambda}\left|\mathcal{T}'_{i}\right|\rceil$ time slots in $\mathcal{T}'_{i}$ that also belong to $\mathcal{T}_{4}$; or equivalently, for any $\mathcal{S}\subseteq\mathcal{T}$, and any $\mathcal{T}'_{i}\subseteq\mathcal{S}$ that is not satisfied by the online solution, $\left|\mathcal{T}'_{i}\right|<\lfloor\frac{\lambda}{\lambda-1}\left|\mathcal{S}\cap\mathcal{T}_{4}\right|\rfloor$, and hence $\left|\mathcal{T}_{i}\right|<\lfloor\frac{1}{\lambda-1}\left|\mathcal{S}\cap\mathcal{T}_{4}\right|\rfloor$. Let $\mathcal{T}'_{4}$ be the union of all sets of consecutive time slots that contain $\mathcal{T}_{4}$, and have a cardinality of $\lfloor\frac{\lambda}{\lambda-1}\left|\mathcal{T}_{4}\right|\rfloor-1$. When $\left|\mathcal{T}_{i}\right|<\lfloor\frac{1}{\lambda-1}\left|\mathcal{S}\cap\mathcal{T}_{4}\right|\rfloor$, since at least one type of resource in at least one required time slot is fully occupied, there can be a set of users in a worst case, demanding all resources in $\left|\mathcal{T}'_{4}\right|$ time slots, with $	\mathscr{P}_{i}\left(\boldsymbol{\rho}\right)=\overline{p}$, where $\left|\mathcal{T}'_{4}\right|<2\lfloor\frac{1}{\lambda-1}\left|\mathcal{T}_{4}\right|\rfloor+\left|\mathcal{T}_{4}\right|$. Thus we have the maximum optimal offline total value
	\[
	V_{opt}\left(\boldsymbol{\rho}\mbox{*}\right)<\left|\mathcal{T}'_{4}\right|\left|\mathcal{R}\right|\underline{p}<\left(2\lfloor\frac{1}{\lambda-1}\left|\mathcal{T}_{4}\right|\rfloor+\left|\mathcal{T}_{4}\right|\right)\left|\mathcal{R}\right|\overline{p}.
	\]
	Therefore, for $\boldsymbol{\rho}\mbox{*}\in\Pi_{2}$, the worst-cast competitive ratio
	\begin{equation}\label{eq:Pi_3_CR}
	\begin{split}
	\alpha^{3}< &~ \frac{\frac{\lambda+1}{\lambda-1}\left|\mathcal{R}\right|\left|\mathcal{T}_{4}\right|\overline{p}}{\left|\mathcal{R}\right|\left|\mathcal{T}_{4}\right|\int_{0}^{\eta}P\left(\rho;\beta_{\mathcal{R}}\right)d\rho}\\
	= &~ \frac{\lambda+1}{\left(\lambda-1\right)\eta'}.
	\end{split}
	\end{equation}
\end{proof}

\begin{theorem}\label{thm:CR-multTSlots}
	The worst-cast competitive ratio achieved by our online resource scheduling strategy using pricing function (\ref{eq:P-multiSlots}) is upper bounded by $\frac{\lambda+1}{\lambda-1}\max\{\alpha_{\mathcal{R}}/\eta,1/\eta'\}$, which is a constant with respect to both $\left|\mathcal{R}\right|$ and $\left|\mathcal{T}\right|$. \textcolor{black}{Here, $\alpha_{R}$ is defined by Eq.~\eqref{eq:alpha_1}, \eqref{eq:alpha_2} or \eqref{eq:alpha_3} for $\beta=\beta_{R}$, and $\eta$ is defined as in Assumption~\ref{asmp:balancedMultiRes}.}
\end{theorem}
\begin{proof}
	The theorem follows immediately from Lemmas~\ref{lem:Pi_1_CR}, \ref{lem:Pi_2_CR} and \ref{lem:Pi_3_CR}.
\end{proof}

%% file: sim.tex
\section{Empirical Studies}
\label{sec:sim}
\textcolor{black}{In this section, we evaluate the proposed pricing and scheduling strategies through simulation studies. To simulate realistic cloud computing scenarios, we relax all the assumptions made before, {\em i.e.}, our parameter settings approximate reality rather than following assumptions we used for theoretical analysis.  
We use a Poisson process to model the arrival of users, and set the arrival rate to be between $20$ and $50$ per time slot. Each user requests $5$ time slots on average and $5$ different types of resources at most, if not otherwise specified. Each user demands $1$ to $3$ percent of each type of resource on average,\footnote{We note these percentages are quite large as compared to the practice that a user may use only a very small percentage of the entire capacity of a cloud system. We set such percentages to evaluate performance of our pricing functions in case that Assumption \ref{asmp:smallDemand} is not true.} with different standard deviation for different resource types, ranging from 0.2 percent to 2 percent. We set $\lambda=1.2$ by default. The time horizon of simulations is set to $1000$ time slots, which is large enough compared to the demand of each user. 
The statistics of the random input variables are stationary in all cases except the last one (shown in Fig.~\ref{fig:sw_arf}). The optimal offline total values are obtained by solving problem \eqref{eq:prob3} with constraints (\ref{eq:prob3}c) and (\ref{eq:prob3}d).}

\textcolor{black}{By relaxing the assumptions, we can now optimize the parameters in our pricing functions, {\em e.g.}, $\beta$, $\underline{p}$ and $\overline{p}$, to maximize the average total social welfare.} Specifically, we use pattern search for the optimization: we repeat each experiment for multiple iterations; in the first iteration, we fix the parameters to random estimates; 
  then we add a perturbation (decays with iterations) to each parameter and run the experiment again; a perturbation is retained from one iteration to the next if the total value is improved. 
 In practice, 
 similar probing of parameter values can be done through online learning techniques such as reinforcement learning.

\textcolor{black}{Our theoretical analysis suggests that, under mild assumptions, the worst-case competitive ratio of social welfare is mainly influenced by the total demand level (see Fig.~\ref{fig:priFcn2_cr}), but not by the number of resource types (Theorem~\ref{thm:CR-balancedMultiRes}), nor by the number of requested time slots (Theorem~\ref{thm:CR-multTSlots}). We now investigate the impact of the three factors on the social welfare and competitive ratio, as well as the robustness of the theoretical results, when the assumptions are relaxed.}

To quantify different demand levels, we define the relative total demand as the ratio between the total demand of all potential users and the total resource supply
Fig.~\ref{fig:sw_demand} shows that, the optimal offline total value, $V_{opt}$, increases almost linearly with a slope of $1$, as the total demand increases. At the same time, the online total value, $V_{ol}$, increases with a smaller slope, causing the competitive ratio to increase noticeably from $1.09$ to $1.78$. Although the results exhibit the average system performance (rather than worst-case competitive ratios), it coincides with our worst-case analysis on the scarcity level, $\beta$, where larger $\beta$ leads to a larger competitive ratio (see Fig.~\ref{fig:priFcn2_cr}).
\begin{figure}[!t]
	\begin{centering}
		\includegraphics[width=0.45\textwidth]{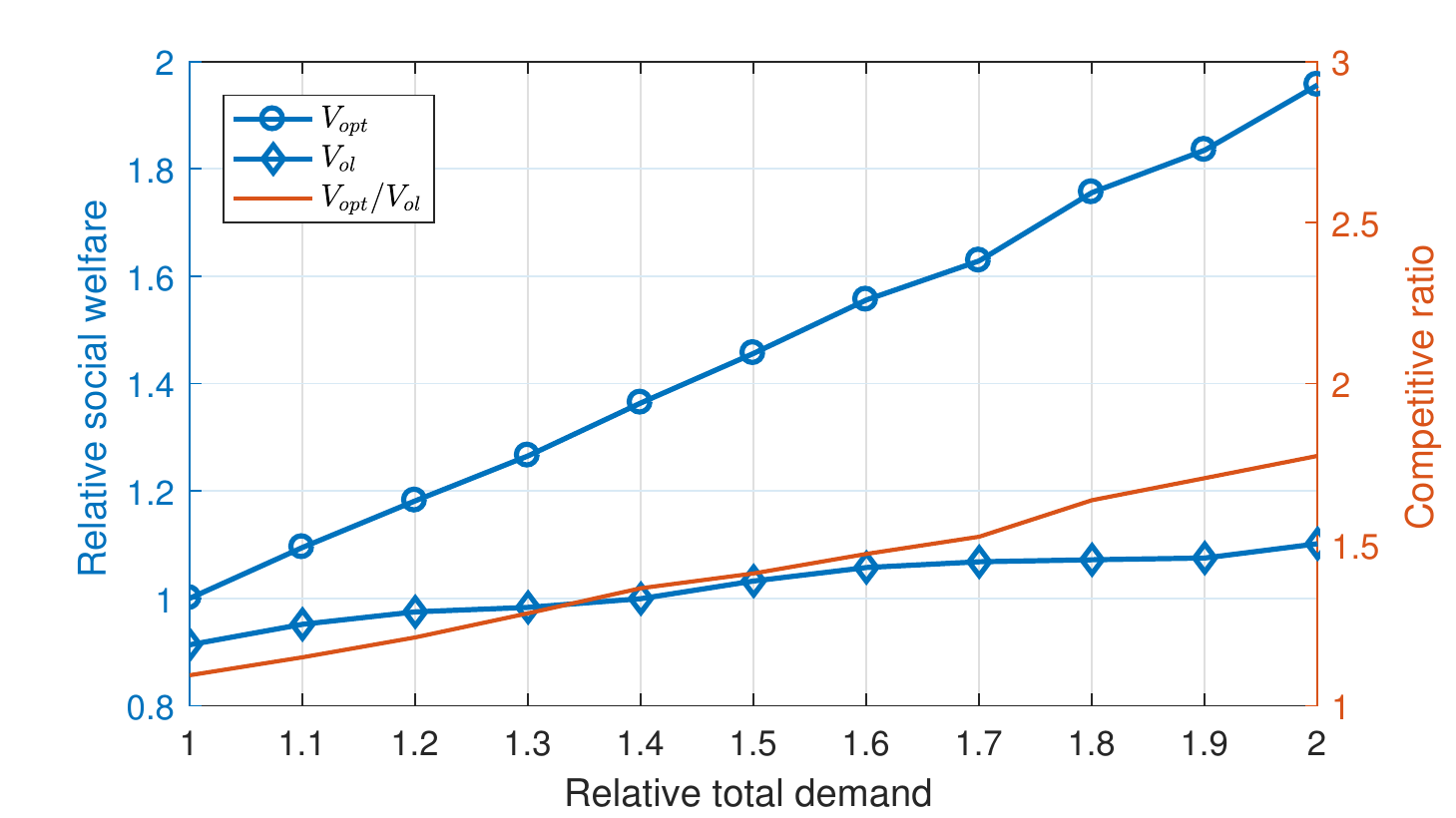}
		\par\end{centering}
	\protect\caption{Online/offline social welfares and competitive ratios given different total demands.}
	\label{fig:sw_demand}
\end{figure}

Next, we vary the number of resource types, $\left|\mathcal{R}\right|$, from $1$ to $10$ to see how it affects the competitive ratio. As shown in Fig.~\ref{fig:sw_nRt}, due to the increase in total demand and total supply, both $V_{opt}$ and $V_{ol}$ increase linearly with the increase of $\left|\mathcal{R}\right|$, while $V_{opt}$ increases slightly faster than $V_{ol}$. Consequently, the competitive ratio only increases mildly (from $1.34$ to $1.57$) as $\left|\mathcal{R}\right|$ increases. The results may indicate that Assumption~\ref{asmp:balancedMultiRes} is slightly violated in practice, since larger $\left|\mathcal{R}\right|$ can increase the chance of unbalanced resource utilization.
\begin{figure}[!t]
\begin{centering}
	\includegraphics[width=0.45\textwidth]{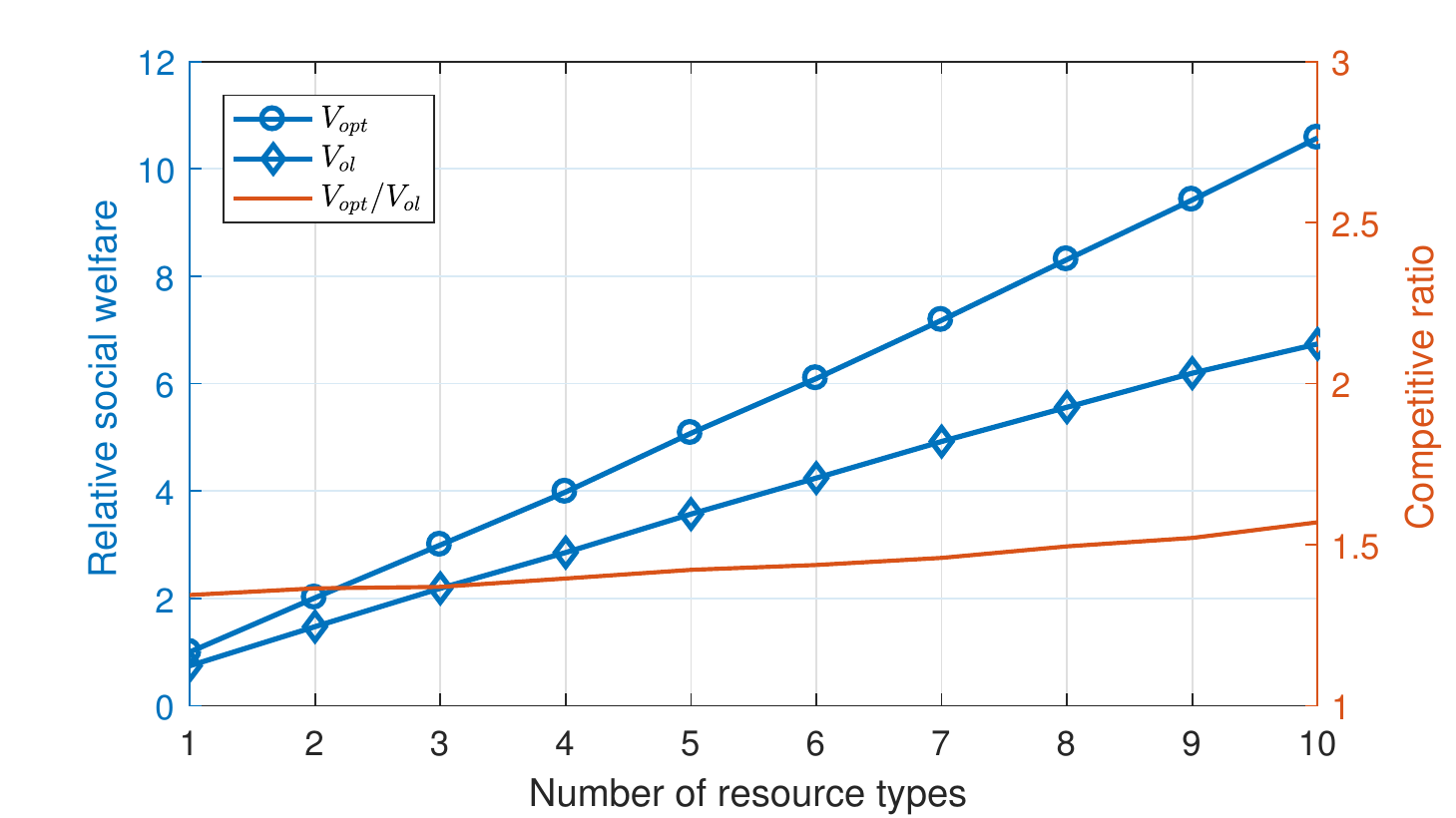}
	\par\end{centering}
\protect\caption{Online/offline social welfares and competitive ratios given different numbers of resource types.}
\label{fig:sw_nRt}
\end{figure}

Similarly, it is also interesting to see how the number of time slots required by each user affects the competitive ratio. Different from the case of varying $\left|\mathcal{R}\right|$, only the total demand will increase with the average number of required time slots. Thus we adjust the demand of each user accordingly to eliminate the effect of increasing relative total demand (see Fig.~\ref{fig:sw_demand}). As we can see in Fig.~\ref{fig:sw_nTs}, $V_{opt}$ and $V_{ol}$ stay almost the same as the average number of required time slots increases, and so does the competitive ratio (varying slightly from $1.41$ to $1.48$). To further verify the proposed strategies, we vary the value of $\lambda$ from $1$ to $1.4$ as shown in Fig.~\ref{fig:sw_lambda}. We test the performance for two different demand levels, with relative total demands of $1.5$ and $3$, respectively. In this case, $V_{opt}$ stays almost the same as $\lambda$ changes and is omitted from the figure. Clearly, $V_{ol}$ increases more from $\lambda=1$ to $\lambda=1.2$ than from $\lambda=1.2$ to $\lambda=1.4$, indicating $\lambda=1.2$ is a good trade-off between the availability and timeliness of service.
\begin{figure}[!t]
	\begin{centering}
		\subfloat[Online/offline social welfares and competitive ratios given different numbers of required time slots.]{
			\includegraphics[width=0.45\textwidth]{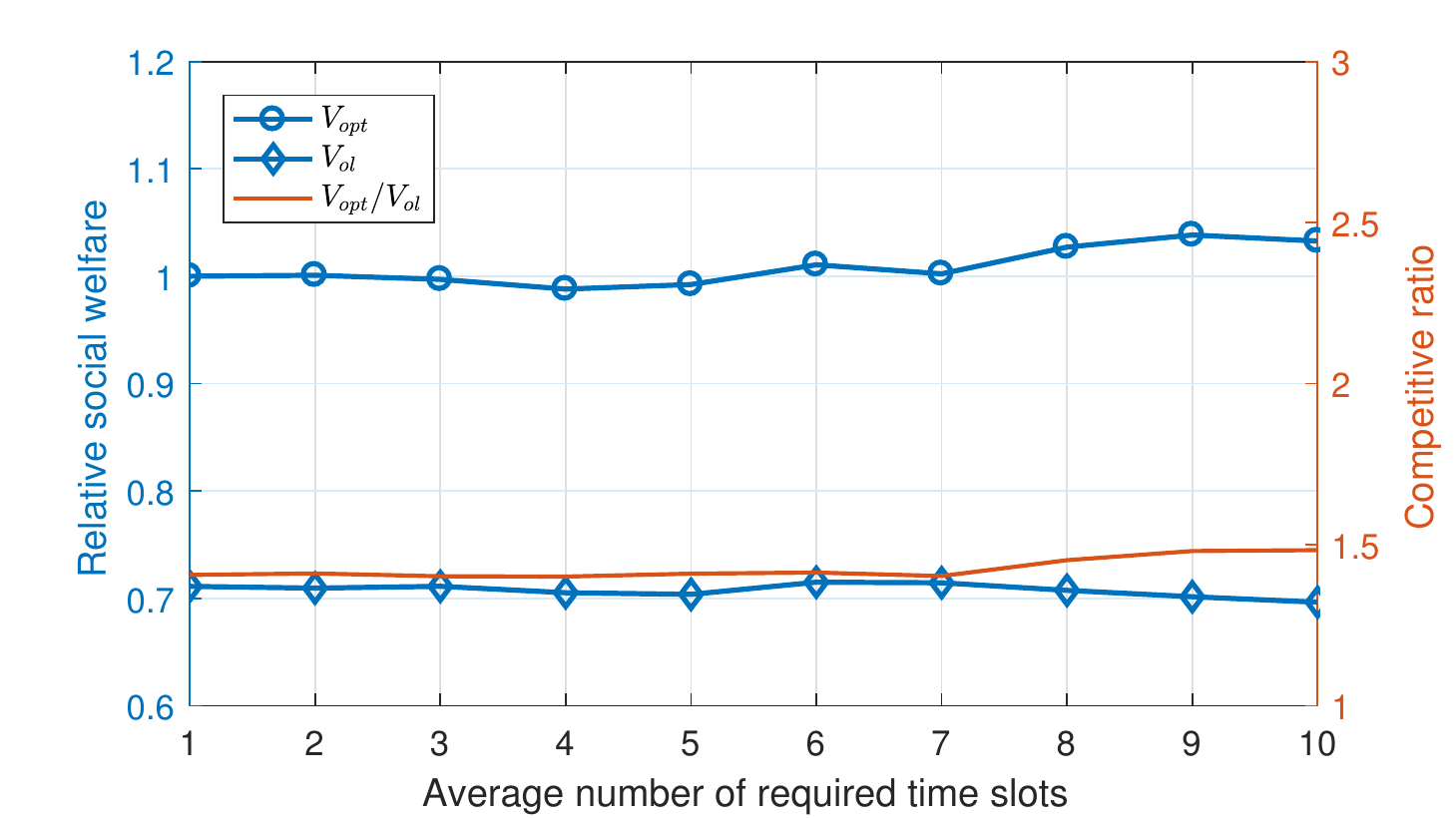}
			\label{fig:sw_nTs}}
		\par\end{centering}
	\vspace*{-10pt}
	\begin{centering}
		\subfloat[Online social welfare at different values of $\lambda$.]{
			\includegraphics[width=0.45\textwidth]{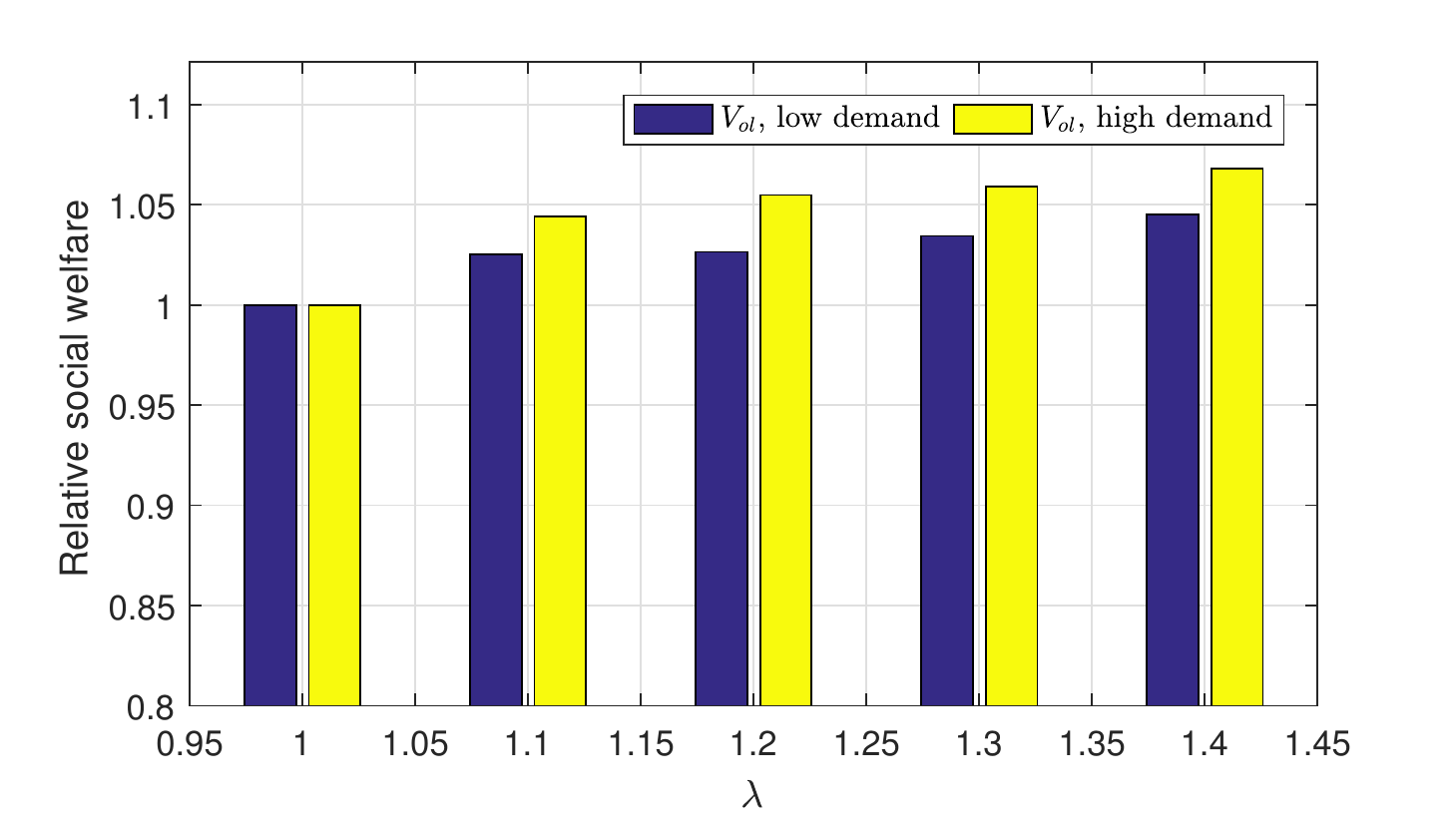}
			\label{fig:sw_lambda}}
		\par\end{centering}
	\vspace*{0pt}
	\caption{Performance of the elastic scheduling strategy discussed in Sec.~\ref{sec:multi_timeslots}.}
	\label{fig:nTs_lambda}
\end{figure}

The simulations conducted so far are based on stationary arrival processes of users. In practice, however, the arrival rate may change over time ({\em e.g.}, fluctuating periodically). To capture this characteristic, we vary the arrival rate according to a sine curve with a period of $100$ time slots. In Fig.~\ref{fig:sw_arf}, as we increase the amplitude of the sine curve from $0$ to $1$ (normalized by the average arrival rate), both $V_{opt}$ and $V_{ol}$ decrease significantly, while the competitive ratio remains at around $1.5$. The reason behind the results is that, when the arrival rate is very low, the resource utilization ratios stay low, so that almost all demands can be satisfied; while when the arrival rate is very high, a high proportion of the demands cannot be satisfied by either the optimal offline solution or the online solution.
\begin{figure}[!t]
\begin{centering}
	\includegraphics[width=0.45\textwidth]{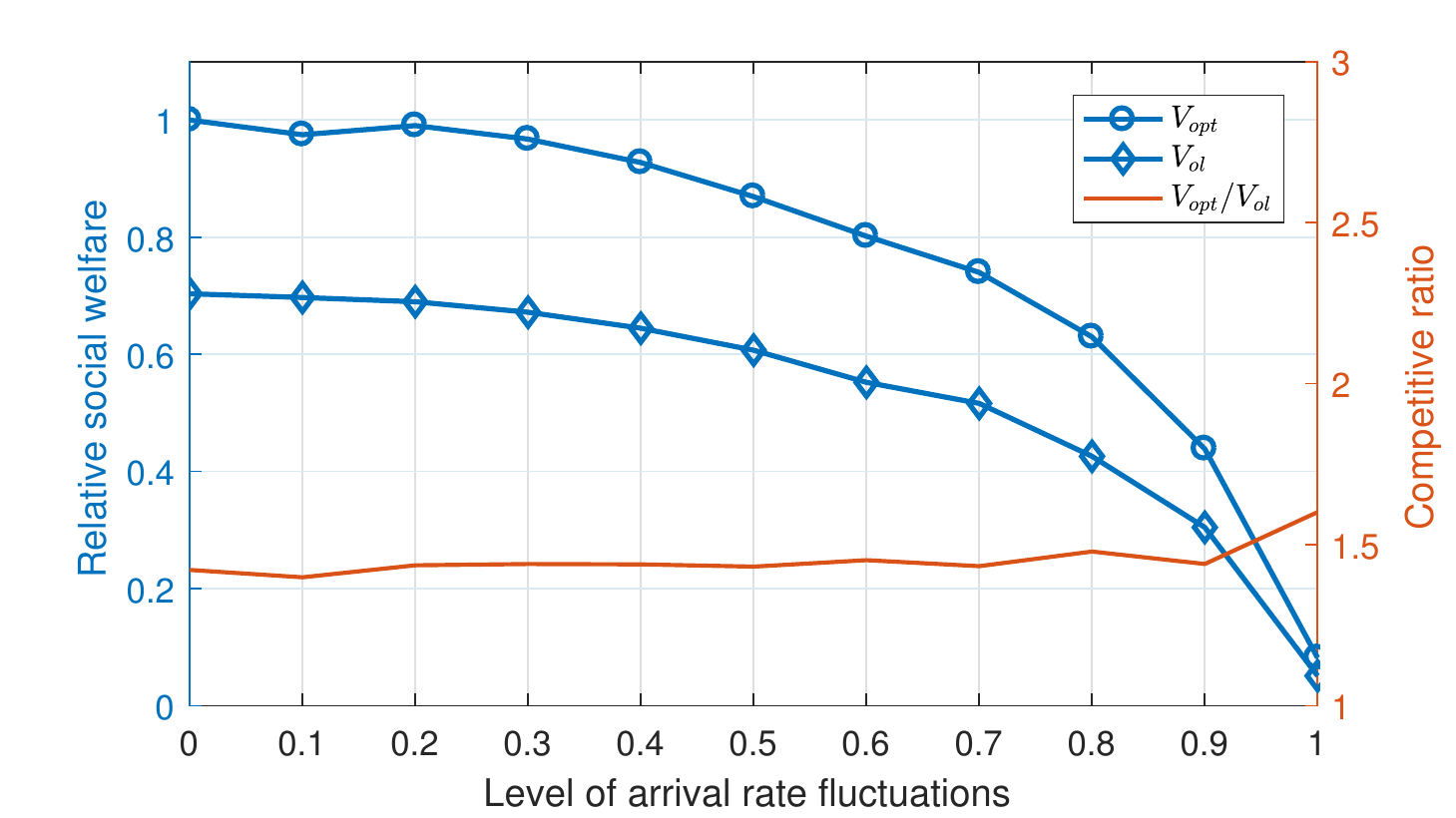}
	\par\end{centering}
\protect\caption{Online/offline social welfare and competitive ratios given different levels of arrival rate fluctuations.}
\label{fig:sw_arf}
\end{figure}

%% file: conclusion.tex
\section{Concluding Remarks}
\label{sec:conclusion}

This paper studies online posted pricing strategies in a number of cloud resource allocation scenarios. We start by investigating the basic case of a single type of cloud resource without resource recycling, and prove optimality of a set of exponential pricing functions in terms of social welfare, which compute unit resource prices based on realtime demand-supply of cloud resources. Exploiting the insights acquired, we further derive pricing functions in practical scenarios with multiple resource types and limited resource occupation durations, and prove tight competitive ratio bounds achieved using these functions, without relying on any particular user arrival process or valuation distribution. Relaxing assumptions made in theoretical analysis, empirical studies further reveal good performance of our pricing functions under realistic settings. Though set up in a cloud computing environment, our models and algorithms are also applicable to posted pricing in other related online resource allocation problems.

%% file: appendix.tex
\appendix

\section{Proof of Claim \getrefnumber{clm:Omega_2_wc}}
\label{sec:clm:Omega_2_wc}
\begin{proof}
	The worst case of online solution is that the valuations of satisfied users are the same as the prices they accept. Thus by Assumption~\ref{asmp:smallDemand}, we have
	\begin{equation}\label{eq:V_ol-Omega_2}
	\begin{split}
	V_{ol}\left(\boldsymbol{\rho}\mbox{*}\right)= & \sum_{r\in\mathcal{R}}\int_{0}^{\rho_{r}\mbox{*}}P\left(\rho;\beta_{r}\right)d\rho\\
	= & \sum_{r\in\mathcal{R}_{1}}\rho_{r}\mbox{*}\underline{p}+\sum_{r\in\mathcal{R}_{2}}\int_{0}^{\rho_{r}\mbox{*}}P\left(\rho;\beta_{r}\right)d\rho,
	\end{split}
	\end{equation}
	as the minimum total value of the online solution. On the other hand, any unsatisfied user $i$ has an average unit value less than $\mathscr{P}_{i}\left(\boldsymbol{\rho}\mbox{*}\right)$, because otherwise $\boldsymbol{\rho\mbox{*}}$ cannot be the final resource utilization. We can decompose each user's value as $v_{i}=\sum_{r\in\mathcal{R}}d_{i,r}U_{i,r}\left(\boldsymbol{\rho}\right)$, and
	\[
	U_{i,r}\left(\boldsymbol{\rho}\right)=\frac{v_{i}}{d_{i}\mathscr{P}_{i}\left(\boldsymbol{\rho}\right)}P\left(\rho_{r};\beta_{r}\right),
	\]
	such that a user $i$'s average unite value $v_{i}/d_{i}<\mathscr{P}_{i}\left(\boldsymbol{\rho}\mbox{*}\right)$ if and only if $U_{i,r}\left(\boldsymbol{\rho}\mbox{*}\right)<P\left(\rho_{r};\beta_{r}\right)$, for any $r\in\mathcal{R}$. Here, $U_{i,r}\left(\boldsymbol{\rho}\mbox{*}\right)$ can be seen as user $i$'s unit value of resource $r$ given a certain $\boldsymbol{\rho}\mbox{*}$.
	
	For $r\in\mathcal{R}_{1}$, in the worst case, there can be a set of unsatisfied users with a total demand of $\min\left\{1,1+\beta_{r}\right\}$ for each type of resource, and with a unit value $U_{i,r}\left(\boldsymbol{\rho}\mbox{*}\right)=\underline{p}-\epsilon_{r}$. Note that $U_{i,r}\left(\boldsymbol{\rho}\mbox{*}\right)<\underline{p}$ does not contradict with Assumption.~\ref{asmp:varCstr}, since a small enough $\epsilon_{r}$ can ensure $v_{i}/d_{i}\geq\underline{p}$. For $r\in\mathcal{R}_{2}$, the discussion on Eq.~\eqref{eq:V_ol1}, \eqref{eq:V_ol2} for a single resource type is still valid if we consider $U_{i,r}\left(\boldsymbol{\rho}\mbox{*}\right)$ as unit value of resource; and according to Eq.~\eqref{eq:stableCR}, we have
	\[
	V_{opt}\left(\rho_{r}\mbox{*}\right)=\alpha_{r}V_{ol}\left(\rho_{r}\mbox{*}\right)-\epsilon_{r}=\alpha_{r}\int_{0}^{\rho_{r}\mbox{*}}P\left(\rho;\beta_{r}\right)d\rho)-\epsilon_{r}.
	\]
	This yields the maximum optimal offline total value given Eq.~\eqref{eq:V_ol-Omega_2}:
	\begin{equation}\label{eq:V_opt-Omega_2}
	\begin{split}
	V_{opt}\left(\boldsymbol{\rho}\mbox{*}\right)= & \sum_{r\in\mathcal{R}_{1}}\underline{p}\min\left\{ 1,1+\beta_{r}\right\} \\
	+ & \sum_{r\in\mathcal{R}_{2}}\alpha_{r}\int_{0}^{\rho_{r}\mbox{*}}P\left(\rho;\beta_{r}\right)d\rho-\epsilon.
	\end{split}
	\end{equation}
	
	For $r\in\mathcal{R}_{1}$, $\rho_{r}\mbox{*}$ only affects the first term of Eq.~\eqref{eq:V_ol-Omega_2}, while the first term of Eq.~\eqref{eq:V_opt-Omega_2} is a constant with respect to $\rho_{r}\mbox{*}$. Thus in any worst case, the first term of Eq.~\eqref{eq:V_ol-Omega_2} should be minimized, and hence $\rho_{r}\mbox{*}=0,\forall r\in\mathcal{R}_{1}$. For $r\in\mathcal{R}_{2}$, let $V_{ol}\left(\rho_{r}\mbox{*}\right)=\int_{0}^{\rho_{r}\mbox{*}}P\left(\rho;\beta_{r}\right)d\rho$, we have
	\begin{align*}
	\alpha\left(\boldsymbol{\rho}\mbox{*}\right)= & \frac{\sup_{\epsilon>0}V_{opt}\left(\boldsymbol{\rho}\mbox{*}\right)}{V_{ol}\left(\boldsymbol{\rho}\mbox{*}\right)}\geq\frac{\sum_{r\in\mathcal{R}_{2}}\alpha_{r}V_{ol}\left(\rho_{r}\mbox{*}\right)}{\sum_{r\in\mathcal{R}_{2}}V_{ol}\left(\rho_{r}\mbox{*}\right)}\\
	\geq & \frac{\alpha_{\underline{r}}\sum_{r\in\mathcal{R}_{2}}V_{ol}\left(\rho_{r}\mbox{*}\right)}{\sum_{r\in\mathcal{R}_{2}}V_{ol}\left(\rho_{r}\mbox{*}\right)}=\alpha_{\underline{r}},
	\end{align*}
	where $\underline{r}=\argmin_{r\in\mathcal{R}_{2}}\alpha_{r}$. When $\left|\mathcal{R}_{2}\right|\geq 2$, we can iteratively move $\underline{r}$ from $\mathcal{R}_{2}$ to $\mathcal{R}_{1}$, and set $\rho_{\underline{r}}\mbox{*}=0$ without decreasing $\alpha\left(\boldsymbol{\rho}\mbox{*}\right)$, until $\left|\mathcal{R}_{2}\right|=1$, since
	\[
	\begin{split}
	\frac{\sup_{\epsilon>0}\left(V_{opt}\left(\boldsymbol{\rho}\mbox{*}\right)-\alpha_{\underline{r}}V_{ol}\left(\rho_{r}\mbox{*}\right)-\epsilon+\underline{p}\min\left\{ 1,1+\beta_{\underline{r}}\right\} \right)}{V_{ol}\left(\boldsymbol{\rho}\mbox{*}\right)-V_{ol}\left(\rho_{r}\mbox{*}\right)}\\
	\geq\frac{\sup_{\epsilon>0}V_{opt}\left(\boldsymbol{\rho}\mbox{*}\right)}{V_{ol}\left(\boldsymbol{\rho}\mbox{*}\right)}.
	\end{split}
	\]
	Similarly, for the only $r\in\mathcal{R}_{2}$, we can decrease $\rho_{r}\mbox{*}$ to $1/\alpha_{r}+\epsilon$ without decreasing $\alpha\left(\boldsymbol{\rho}\mbox{*}\right)$. Therefore, for $\boldsymbol{\rho}\mbox{*}\in\Omega_{2}$, there exists a worst case that happens when $\rho_{r}\mbox{*}=0$ for $r\in\mathcal{R}_{1}$, and $\rho_{r}\mbox{*}=1/\alpha_{r}+\epsilon$ for $r\in\mathcal{R}_{2}$, where $\left|\mathcal{R}_{2}\right|=1$.
\end{proof}

\section{Proof of Claim \getrefnumber{clm:Omega_3_wc}}
\label{sec:clm:Omega_3_wc}
\begin{proof}
	The worst case of online solution is that the valuations of satisfied users are the same as the prices they accept. Thus by Assumption~\ref{asmp:smallDemand}, we have
	\begin{equation}\label{eq:V_ol-Omega_3}
	\begin{split}
	V_{ol}\left(\boldsymbol{\rho}\mbox{*}\right)= & \sum_{r\in\mathcal{R}}\int_{0}^{\rho_{r}\mbox{*}}P\left(\rho;\beta_{r}\right)d\rho\\
	= & \sum_{r\in\mathcal{R}_{3}}\int_{0}^{\rho_{r}\mbox{*}}P\left(\rho;\beta_{r}\right)d\rho+\sum_{r\in\mathcal{R}_{4}}\int_{0}^{1}P\left(\rho;\beta_{r}\right)d\rho,
	\end{split}
	\end{equation}
	as the minimum total value of the online solution. On the other hand, since there is at least one type of resource being fully occupied, i.e., $\left|\mathcal{R}_{4}\right|\geq 1$, there can be a case where all subsequent users demand a small amount of resource $r\in\mathcal{R}_{4}$, making it impossible to satisfy their demands regardless of their valuations. Hence the maximum optimal offline total value
	
	\vspace{-4mm}
	\begin{equation}\label{eq:V_opt-Omega_3}
	\begin{split}
	V_{opt}\left(\boldsymbol{\rho}\mbox{*}\right)= & \sum_{r\in\mathcal{R}_{3}}\int_{\rho_{r}^{1}}^{\rho_{r}^{2}}P\left(\rho;\beta_{r}\right)d\rho\\
	+ & \sum_{r\in\mathcal{R}_{3}}\overline{p}\min\left\{ 1,1+\beta_{r}-\rho_{r}\mbox{*}\right\} \\
	+ & \sum_{r\in\mathcal{R}_{4}}\alpha_{r}\int_{0}^{1}P\left(\rho;\beta_{r}\right)d\rho,
	\end{split}
	\end{equation}
	where $\rho_{r}^{1}=\max\left\{ 0,\beta_{r}\right\}$ and $\rho_{r}^{2}=\max\left\{ \beta_{r},\rho_{r}\mbox{*}\right\}$.
	
	For $r\in\mathcal{R}_{3}$, Eq.~\eqref{eq:V_ol-Omega_3} stays the same or increases as any $\rho_{r}\mbox{*}$ increases, while Eq.~\eqref{eq:V_opt-Omega_3} stays the same or decreases. Thus there exists a worst case where $\rho_{r}\mbox{*}=0,\forall r\in\mathcal{R}_{3}$. Let $\underline{r}=\argmin_{r\in\mathcal{R}_{4}}\alpha_{r}$. Due to the same reason as discussed for Eq.~\eqref{eq:V_ol-Omega_2} and Eq.~\eqref{eq:V_opt-Omega_2}, when $\left|\mathcal{R}_{4}\right|\geq 2$, we can iteratively move $\underline{r}$ from $\mathcal{R}_{4}$ to $\mathcal{R}_{3}$, and set $\rho_{\underline{r}}\mbox{*}=0$ without decreasing the competitive ratio, until $\left|\mathcal{R}_{4}\right|=1$.
\end{proof}